\documentclass[journal,12pt,onecolumn,draftclsnofoot]{IEEEtran}

\usepackage[noadjust]{cite}

\usepackage{url}
\usepackage{amssymb,latexsym,amsfonts,amsmath, amsthm}
\usepackage{graphicx}
\usepackage{eso-pic}
\usepackage{epsfig}
\usepackage{epstopdf}
\usepackage{subfig}
\usepackage{setspace}
\usepackage{comment}
\usepackage[normalem]{ulem}
\usepackage{soul}

\newtheorem{theorem}{Theorem}[section]

\newtheorem{problem}[theorem]{Problem}
\newtheorem{algorithm}[theorem]{Algorithm}
\newtheorem{proposition}[theorem]{Proposition}
\newtheorem{corollary}[theorem]{Corollary}

\newtheorem{definition}[theorem]{Definition}
\newtheorem{example}[theorem]{Example}

\newtheorem{remark}[theorem]{Remark}

\newcommand{\N}{{\mathbb{N}}}
\newcommand{\R}{{\mathbb{R}}}
\newcommand{\Z}{\mathbb{Z}}

\newcommand{\Post}[3]{\mathsf{Post}_{#2}^{#3}(#1)}
\newcommand{\Last}[1]{\mathsf{Last}(#1)}

\newcommand{\Reach}{{\mathsf{Reach}}}
\newcommand{\Safe}{{\mathsf{Safe}}}

\newcommand{\e}{\mathsf{e}}
\DeclareMathOperator{\post}{post}

\def\a{\alpha}
\def\b{\beta}
\def\d{\delta}
\def\A{{\cal A}}

\newcommand{\SCOTS}{{\tt SCOTS}}
\newcommand{\MATLAB}{{\tt MATLAB}}
\newcommand{\CPP}{{\tt C++}}
\newcommand{\ALPAGA}{{\tt ALPAGA}}

\newcommand{\sym}[1]{{#1}_{q}}

\newcommand{\hide}[1]{}

\newcommand{\new}[1]{{#1}}
\newcommand{\moved}[1]{{#1}}
\newcommand{\old}[1]{{\hide{#1}}}

\newcommand{\newagain}[1]{{#1}}
\newcommand{\movedagain}[1]{{#1}}
\newcommand{\oldagain}[1]{{\hide{#1}}}
\newcommand{\oldagainnostrike}[1]{{\hide{#1}}}

\makeatletter
\def\old@comma{,}
\catcode`\,=13
\def,{%
	\ifmmode%
	\old@comma\discretionary{}{}{}%
	\else%
	\old@comma%
	\fi%
}
\makeatother

\title{A Framework for Output-Feedback Symbolic Control}
\author{
	Mahmoud Khaled, Kuize Zhang, and Majid Zamani
	\thanks{
		M. Khaled is with the Department of Computer and Systems Engineering, Faculty of Engineering, Minia University, Egypt.
		K. Zhang is with the Department of Computer Science, University of Surrey, UK.
		M. Zamani is with the Department of Computer Science, University of Colorado Boulder, USA, and the Department of Computer Science at LMU, Munich, Germany.
		Emails: \texttt{mkhaled@mu.edu.eg, kuize.zhang@surrey.ac.uk, majid.zamani@colorado.edu}.
		This work was supported in part by the H2020 ERC Starting Grant AutoCPS (Grant Agreement No. 804639) and the NSF under Grant CNS-2145184.
	}
}

\begin{document}
\maketitle

\begin{abstract}
Symbolic control is an abstraction-based controller synthesis approach that provides, algorithmically, certifiable-by-construction controllers for cyber-physical systems.
Symbolic control approaches usually assume that full-state information is available which is not suitable for many real-world applications with partially-observable states or output information.
This article introduces a framework for output-feedback symbolic control.
We propose relations between original systems and their symbolic models based on outputs.
They enable designing symbolic controllers and refining them to enforce complex requirements on original systems.
We provide example methodologies to synthesize and refine output-feedback symbolic controllers.
\end{abstract}

\linespread{0.90}

\section{Introduction}

In the past decades, the world has witnessed many emerging applications formed by the tight interaction of physical systems, computation platforms, communication networks, and software. 
This is clearly the case in avionics, automotive systems, smart power grids, and infrastructure management systems, which are all examples of so-called cyber-physical systems (CPS).
In CPS, (embedded) control software orchestrates the interaction between different physical and computational parts to achieve some given desired requirements.
Today's CPS often require certifiable control software, faster requirements-to-prototype development cycles and the handling of more sophisticated specifications.
Many CPS are also safety-critical in which the correctness of control software is crucial.
Consequently, modern CPS require approaches for automated synthesis of provably-correct control software. 

Symbolic control \cite{P.Tabuada_VCHS_Symbolic,M.Zamani_etal_SC_Nonlin_NoStabilAssump,MZ,G.Reissig_etal_FRR_TAC} is an approach to automatically synthesize certifiable controllers that handle complex requirements including objectives and constraints given by formulae in linear temporal logic (LTL) or automata on infinite strings \cite{P.Tabuada_VCHS_Symbolic,C.Baier_PrincipModelChecking}.
\new{In symbolic control, a}\old{A given} dynamical system (e.g., a physical process described by a set of differential equations) is related to a \emph{symbolic model} (i.e., a system with finite state and input sets) via a formal relation.
The relation ensures that the symbolic model captures some required features from the original system. 
\oldagain{As}\newagain{Since} symbolic models are finite, reactive synthesis techniques \cite{Pnueli_1989,Vardi1995,BLOEM2012911} \oldagain{are applicable}\newagain{can be applied} to algorithmically synthesize controllers enforcing the given specifications.
The designed controllers are usually referred to as \emph{symbolic controllers}.

Symbolic models can be used to abstract several classes of control systems \cite{P.Tabuada_VCHS_Symbolic,MZ,G.Reissig_etal_FRR_TAC,mkhaled_allerton16,zamani}.
They have been recently investigated for general nonlinear systems \cite{M.Zamani_etal_SC_Nonlin_NoStabilAssump,G_POLA_Automatica_2008}, time-delay control systems \cite{GPola_TimeDelay_2010}, switched control systems \cite{MZamani_SwitchedSystems_Automatica2015,AGirard_SwitchedSystems_TAC2010}, stochastic control systems \cite{MZamani_StochasticSystems_TAC14,MZamani_StochasticSystems_ECC13}, and networked control systems \cite{}.
Unfortunately, the majority of current techniques assume control systems with full-state or quantized-state information and, hence, they are not applicable to control systems with outputs or partially-observable states\oldagain{, which is the case in many practical applications}.
Moreover, \oldagain{all}\newagain{none of} state-of-the-art tools of symbolic controller synthesis \cite{SCOTS,CoSyMA,pFaces} \oldagain{do not }support output-feedback systems since the required theories for them are not yet fully established.

In this article, we consider control systems with partial-state or output information.
We refer to these particular types of systems as \emph{output-based control systems}.
We introduce a framework for symbolic control that can handle this class of systems.
We refer to the introduced framework as \emph{output-feedback symbolic control}.
We first extend the work in \cite{G.Reissig_etal_FRR_TAC} to provide mathematical tools for constructing symbolic models of output-based systems.
More precisely, output-feedback refinement relations (OFRRs) are introduced as means of relating output-based systems and their symbolic models.
They are extensions of feedback refinement relations (FRRs) in \cite{G.Reissig_etal_FRR_TAC}.
\newagain{OFRRs allow abstractions to be constructed by quantizing
the state and output sets of concrete systems, such that the output quantization respects the state quantization.}
We \oldagain{show}\newagain{prove} that OFRRs ensure external (i.e., output-based) behavioral inclusion from original systems to symbolic models.
Symbolic controllers synthesized based on the outputs of symbolic models can be refined via simple and practically implementable interfaces. 
\old{Then, we present three methodologies of output-feedback symbolic control where general, possibly nonlinear, output-based control systems are considered.
The first methodology is based on games of imperfect information.
The second one proposes designing observers for output-based systems.
The third one proposes detectors designed for symbolic models.
Examples are presented to demonstrate the effectiveness of proposed methodologies.}

In Sections \ref{meth_game_based}, \ref{mthd_concrete_domain_observers} and \ref{meth_detector}, we present example methodologies \oldagain{of output-feedback symbolic control }that \oldagain{utilizes}\newagain{realize} the introduced framework.
The first methodology is based on games of imperfect information.
The second one proposes designing observers for output-based systems.
The third one proposes detectors designed for symbolic models.
Three case studies are presented in Section \ref{sec_examples} to demonstrate the effectiveness of proposed methodologies.

\section{Notation}
\label{sec_notation}
The identity map on a set $X$ is denoted by $id_X$. 
Symbols $\N, \Z, \R, \R^+$, and $\R^+_0$ denote, respectively, the sets of natural, integer, real, positive real, and nonnegative real numbers. 

The relative complement of a set $A$ in a set $B$ is denoted by $B \backslash A$.
For a set $A$, we denote by $\vert A \vert$ the cardinality of the set, and by $2^A$ the set of all subsets of $A$ including the empty set $\emptyset$.
A \emph{cover} of a set $A$ is a set of subsets of $A$ whose union equals $A$.
A \emph{partition} of a set $A$ is a set of pairwise disjoint \newagain{nonempty} subsets of $A$ whose union equals $A$.
We denote by $A^*$ the set of all finite strings (a.k.a. sequences) obtained by concatenating elements in $A$, by $A^\omega$ the set of all infinite strings obtained by concatenating elements in $A$, and by $A^\infty$ the set of all finite and infinite strings obtained by concatenating elements in $A$.
For any finite string $s$, $\vert s \vert$ denotes the length of the string, $s_i$, $i \in \{0,1,\cdots,\vert s \vert-1\}$, denotes the $i$-th element of $s$, and $s[i,j]$, $j \geq i$, denotes the substring $s_i s_{i+1} \cdots s_{j}$.
Symbol $\e$ denotes the empty string and $\vert \e \vert = 0$.
We use the dot symbol $\cdot$ to concatenate two strings.

Consider a relation $\mathcal{R} \subseteq A \times B$.
$\mathcal{R}$ is strict when $\mathcal{R}(a) \neq \emptyset$ for every $a \in A$.
$\mathcal{R}$ naturally introduces a map $\mathcal{R}: A \to 2^B$ such that $\mathcal{R}(a) =\{b \in B \;\vert\; (a,b) \in \mathcal{R}\}$.
$\mathcal{R}$ also admits an inverse relation $\mathcal{R}^{-1} := \{ (b,a) \in B \times A \;\vert\; (a,b) \in \mathcal{R} \}$.
Given an element $r = (a,b) \in \mathcal{R}$, $\pi_A(r)$ denotes the natural projection of $r$ on the set $A$, i.e., $\pi_A(r) = a$.
We sometimes abuse the notation and apply the projection map $\pi_A$ to a string (resp., a set of strings) of elements of $\mathcal{R}$, which means applying it iteratively to all elements in the string (resp., all strings in the set).
When $\mathcal{R}$ is an equivalence relation on a set $X$, we denote by $[x]$ the equivalence class of $x \in X$ and by $X/\mathcal{R}$ the set of all equivalence classes (a.k.a. quotient set). 
We also denote by $\pi_\mathcal{R} : X \to X/\mathcal{R}$ the natural projection map taking a point $x \in X$ to its equivalence class, i.e., $\pi_\mathcal{R}(x) = [x] \in X/\mathcal{R}$. 
We say that an equivalence relation is finite when it has finitely many equivalence classes.

Given a vector $v \in \R^n$, we denote by $v_i$, $i \in \{0,1, \cdots, n-1\}$, the $i$-th element of $v$ and by $\Vert v \Vert$ its infinity norm. 

\section{Preliminaries}
\label{sec_definitions}
First, we present the notion of \emph{systems} as a general mathematical framework to describe control systems, symbolic models, observers, controllers, and their interconnections. 

\subsection{Systems}
\label{systems}
We use a similar definition for systems as in \cite{P.Tabuada_VCHS_Symbolic}.
\begin{definition}[System]
	\label{def_system}
	A system is a tuple 
	\begin{equation}
		\nonumber
		S := (X, X_0, U, \longrightarrow, Y, H),
	\end{equation}
	where $X$ is the set of states,
	$X_0 \subseteq X$ is a set of initial states,
	$U$ is the set of inputs, 
	$\longrightarrow  \subseteq X \times U \times X$ is \oldagain{a}\newagain{the} transition relation,
	$Y$ is the set of outputs, and 
	$H: X \to Y$ is \oldagain{an}\newagain{the} output map.
\end{definition}

All sets in tuple $S$ are assumed to be non-empty.
For any $x \in X$ and $u \in U$, we denote by $\Post{x}{u}{S} := \{x' \in X \;\vert\; (x,u,x') \in \longrightarrow \}$ the set of $u$-successors of $x$ in $S$.
When $S$ is known from the context, the set of $u$-successors of $x$ is simply denoted by $\Post{x}{u}{}$.
The inputs admissible to a state $x$ of system $S$ is denoted by $U_{S}(x) := \{ u \in U \;\vert\; \Post{x}{u}{} \neq \emptyset \}$.

For any output element $y \in Y$, the map $H^{-1} : Y \to 2^X$ recovers the underlying set of states $X_y \subseteq X$ generating $y$, and it is defined as follows: $H^{-1}(y) := \{ x \in X \;\vert\; H(x) = y \}$.

We sometimes abuse the notation and apply maps $H$ and $H^{-1}$ to subsets of $X$ and $Y$, respectively, which refers to applying them element-wise and then taking the union. Specifically, we have that

\begin{equation}
	\nonumber
	\text{for } \bar{x} \subseteq X \text{, }  H(\bar{x}) := \underset{x \in \bar{x}}{\bigcup} \{ H(x) \}\text{, and}
\end{equation}
\begin{equation}
	\nonumber
	\text{for } \bar{y} \subseteq Y \text{, }  H^{-1}(\bar{y}) := \underset{y \in \bar{y}}{\bigcup} H^{-1}(y).
\end{equation}

System $S$ is said to be 
static if $X$ is singleton;
autonomous if $U$ is singleton;
state-based (a.k.a. simple system \cite{G.Reissig_etal_FRR_TAC}) when $X = Y$, $H = id_X$, and all states are admissible as initial ones, i.e., $X = X_0$;
output-based when $X \neq Y$;
total when for any $x \in X$ and any $u \in U$ there exists at least one $x' \in X$ such that $x' \in \Post{x}{u}{}$;
deterministic when for any $x \in X$ and any $u \in U$ we have $\vert \Post{x}{u}{} \vert \leq 1$; and
symbolic when $X$ and $U$ are both finite sets.

For any $\bar{x} \subseteq X_0$, we denote by $S^{(\bar{x})}$ the restricted version of $S$ with $X_0 = \bar{x}$.
For any output-based system $S$, one can always construct its state-based version by assuming the availability of state information, i.e., $Y = X$, $X_0 = X$ and $H = id_X$, and we denote it by $S_{X}$.

\new{Let $S$ be an output-based system. }Map $\bar{U}_S: Y \to 2^U$ provides all inputs admissible to outputs of $S$.
It is defined as follows for any $y \in Y$:
\begin{equation}
	\nonumber
	\bar{U}_S(y) := \underset{x \in H^{-1}(y)}{\bigcap} U_S(x).
\end{equation}
\old{For any $u \in \bar{U}_S(y)$, map $\overline{\mathsf{Post}}_{u}^{S}: Y \to 2^Y$ provides all $u$-successor observations of $S$ such that for any $y \in Y$, $\overline{\mathsf{Post}}_{u}^{S}(y) := H( \underset{x \in H^{-1}(y)}{\bigcup} \Post{x}{u}{S})$.}\new{Additionally, for any $y \in Y$ and $u \in \bar{U}_S(y)$, $\overline{\mathsf{Post}}_{u}^{S}(y)$ denotes all $u$-successor observations of $y$ and we define it as follows:
\begin{equation}
	\nonumber
	\overline{\mathsf{Post}}_{u}^{S}(y) := H( \underset{x \in H^{-1}(y)}{\bigcup} \Post{x}{u}{S}).
\end{equation}}

\moved{Given a system $S$, for all $x\in X$ and $\a\in U^*$ such that $|\a|\ge 1$, $x'\in X$ is called an $\a$-successor of $x$, if there exist states $x_0,\dots,x_{|\a|}\in X$ such that $x_0=x$, $x_{|\a|}=x'$, and $(x_i,\a_i,x_{i+1})\in\longrightarrow$ for all integers $0\le i\le |\a|-1$.
The set of $\a$-successors of a state $x\in X$ (resp., a subset $X'\subset X$) is denoted by $\Post{x}{\a}{}$ (resp., $\Post{X'}{\a}{}:=\cup_{x\in X'}\Post{x}{\a}{}$).
For all $x\in X$, $\a\in U^*$ and $\b\in Y^*$ such that $|\a|+1=|\b|$, $x'\in X$ is called an $(\a,\b)$-successor of $x$, if there exist states $x_0,\dots,x_{|\a|}\in X$ such that $x_0=x$, $x_{|\a|}=x'$, $H(x_{|\a|})=\b_{|\a|}$, and $H(x_i)=\b_i$ and $(x_i,\a_i,x_{i+1})\in\longrightarrow$ for all integers $0\le i\le |\a|-1$.
The set of $(\a,\b)$-successors of a state $x\in X$ (resp., a subset $X'\subset X$) is denoted by $\Post{x}{\a}{\b}$ (resp., $\Post{X'}{\a}{\b}:=\cup_{x\in X'}\Post{x}{\a}{\b}$).}

An \emph{internal run} of system $S$ is an infinite sequence $r_{int} := x_0 u_0 x_1 u_1 \cdots x_{n-1} u_{n-1} x_n \cdots$ such that $x_0 \in X_0$, and for any $i \ge 0$ we have $(x_i,u_i,x_{i+1}) \in \longrightarrow$.
An \emph{external run} is an infinite sequence $r_{ext} := y_0 u_0 y_1 \cdots y_{n-1} u_{n-1} y_n \cdots$ such that $y_0 = H(x_0)$ for some $x_0 \in X_0$, and 
for any $i \ge 0$ there exist $x_i \in X$ and $x_{i+1} \in X$ such that $y_i = H(x_i)$, $y_{i+1} = H(x_{i+1})$, and $(x_i,u_i,x_{i+1}) \in \longrightarrow$.
The internal (resp., external) \emph{prefix} up to $x_n$ (resp., $y_n$) of $r_{int}$ (resp., $r_{ext}$) is denoted by $r_{int}(n)$ (resp., $r_{ext}(n)$) and its last element is $\Last{r_{int}(n)} := x_n$ (resp., $\Last{r_{ext}(n)} := y_n$).
The set of all internal (resp., external) runs and the set of all internal (resp., external) $n$-length prefixes are denoted by $\mathsf{RUNS}_{int}(S)$ (resp., $\mathsf{RUNS}_{ext}(S)$) and $\mathsf{PREFS}^n_{int}(S)$ (resp., $\mathsf{PREFS}^n_{ext}(S)$), respectively.
A state $x$ is said to be \emph{reachable} iff there exists at least one internal \emph{prefix} $r_{int}(n) \in \mathsf{PREFS}^n_{int}(S)$ such that $\Last{r_{int}(n)} = x$ for some $n \in \N$.

\subsection{Composition of systems}
Systems are composed together to construct new systems.
Here, we define formally different types of compositions.

\begin{definition}[Serial Composition]
	\label{def_serial_composition}
	Consider two systems $S_i := (X_i, X_{i,0}, U_i, \underset{i}{\longrightarrow}, Y_i, H_i)$, $i \in \{1,2\}$, such that $Y_1 \subseteq U_2$. 
	The serial (a.k.a. cascade) composition of $S_1$ and $S_2$, denoted by $S_2 \circ S_1$, is a new system $S_{12} := (X_1 \times X_2, X_{1,0} \times X_{2,0}, U_1, \underset{12}{\longrightarrow}, Y_2, H_{12})$, where $((x_1,x_2),u_1,(x'_1, x'_2)) \in \underset{12}{\longrightarrow}$ iff there exist two transitions $(x_1, u_1, x'_1) \in \underset{1}{\longrightarrow}$ and $(x_2, H_1(x_1), x'_2) \in \underset{2}{\longrightarrow}$, and map $H_{12}$ is defined as follows for any $(x_1, x_2) \in X_1 \times X_2$: $H_{12}((x_{1},x_{2})) := H_2(x_2)$.
\end{definition}

\begin{definition}[Feedback Composition]
	\label{def_feedback_composition}
	Consider two systems $S_i := (X_i, X_{i,0}, U_i, \underset{i}{\longrightarrow}, Y_i, H_i)$, $i \in \{1,2\}$, such that $Y_1 \subseteq U_2$, $Y_2 \subseteq U_1$, and the following holds:
	\begin{align}
		\nonumber
		y_2 = H_2(x_2) \;\land\; y_1 = H_1(x_1) \;\land\; \Post{x_1}{y_2}{S_1} = \emptyset \\
		\nonumber
		\implies \Post{x_2}{y_1}{S_2} = \emptyset.
	\end{align}
	Then, $S_1$ is said to be \emph{feedback-composable} with $S_2$ (denoted by $S_1 \times S_2$) and the new composed system is
	$S_{12} := (X_1 \times X_2, X_{1,0} \times X_{2,0}, \{0\}, \underset{12}{\longrightarrow}, Y_1 \times Y_2, H_{12})$, where $((x_1,x_2),0,(x'_1, x'_2)) \in \underset{12}{\longrightarrow}$ iff there exist two transitions $(x_1, H_2(x_2), x'_1) \in \underset{1}{\longrightarrow}$ and $(x_2, H_1(x_1), x'_2) \in \underset{2}{\longrightarrow}$, and the map $H_{12}$ is defined as follows for any $(x_1, x_2) \in X_1 \times X_2$:
	\begin{equation}
		\nonumber
		H_{12}((x_1,x_2))  := (H_1(x_1), H_2(x_2)).
	\end{equation}
\end{definition}

The feedback composition in \cite{G.Reissig_etal_FRR_TAC} requires that one of the systems is Moore (i.e., the output does not depend on the input).
Such assumption is already fulfilled here since all systems are Moore by Definition \ref{def_system}.
The following proposition shows that external runs of feedback-composed systems are tightly connected to external runs of their subsystems.
It is used later in Subsection \ref{ssec_controller_synthesis} to prove the output-based behavioral inclusion from original systems to symbolic models.

\begin{proposition}
	\label{prop_FB_runs}
	Consider two systems $S_i := (X_i, X_{i,0}, U_i, \underset{i}{\longrightarrow}, Y_i, H_i)$, $i \in \{1,2\}$, such that $S_1$ is feedback-composable with system $S_2$.
	Then, for a feedback-composed system $S_1 \times S_2$, an external run
	$r_{S_1 \times S_2, ext} := (y_{1,0},y_{2,0})0(y_{1,1},y_{2,1})0 \cdots$
	exists iff there exist two external runs
	$r_{S_1 , ext} := y_{1,0}y_{2,0}y_{1,1}y_{2,1} \cdots$ and
	$r_{S_2 , ext} := y_{2,0}y_{1,0}y_{2,1}y_{1,1} \cdots$.
\end{proposition}
\begin{proof}
	\label{proof_prop_FB_runs}
	The proof is straightforward based on Definitions \ref{def_system} and \ref{def_feedback_composition}.
\end{proof}

\begin{definition}[Observation Composition] 
	\label{def_observ_composition} 	
	Consider two systems $S_i := (X_i, X_{i,0}, U_i, \underset{i}{\longrightarrow}, Y_i, H_i)$, $i \in \{1,2\}$, such that $U_1 \times Y_1 \subseteq U_2$. 
	The observation composition of $S_1$ and $S_2$, denoted by $S_2 \triangleleft S_1$, is a new system $S_{12} := (X_1 \times X_2, X_{1,0} \times X_{2,0}, U_1, \underset{12}{\longrightarrow}, X_2, H_{12})$, where $((x_1,x_2),u_1,(x'_1, x'_2)) \in \underset{12}{\longrightarrow}$ iff there exist two transitions: $(x_1, u_1, x'_1) \in \underset{1}{\longrightarrow}$ and $(x_2, (u_1,H_1(x_1)), x'_2) \in \underset{2}{\longrightarrow}$, and $H_{12} := \pi_{X_2}$. 
\end{definition}

The observation composition is used when system $S_2$ is an observer that infers the states of $S_1$ by monitoring its inputs and outputs.

\subsection{Specifications and Control Problems}
Now, we discuss the behaviors of systems and their specifications.
Let $S$ be a system as defined in Definition \ref{def_system}.
The internal and external behaviors of $S$ are subsets of the set of all (possibly infinite) internal and external prefixes of $S$, i.e., $B_{int}(S) \subseteq \bigcup_{n\in \N \cup \{ \infty \}}  \mathsf{PREFS}^n_{int}(S)$ and $B_{ext}(S) \subseteq \bigcup_{n\in \N \cup \{ \infty \}}  \mathsf{PREFS}^n_{ext}(S)$. Specifications are defined next.

\begin{definition}[Specification]
	\label{def_specification}
	Let $S$ be a system as defined in Definition \ref{def_system}.
	Let $\Gamma_S := \pi_{Y}(B_{ext}(S))$ be the set of all output sequences of $S$.
	A \emph{specification} $\psi \subseteq \Gamma_S$ is a set of output sequences that must be enforced on $S$.
	System $S$ satisfies $\psi$ (denoted by $S \models \psi$) iff $\pi_{Y}(B_{ext}(S)) \subseteq \psi$.
\end{definition}

Specifications can adopt formal requirements encoded as linear temporal logic (LTL) \cite{Pnueli77} formulae or automata on finite strings.
Classical requirements like \emph{invariance} (often referred to as \emph{safety}) and \emph{reachability} can be readily included.
Given a safe set of observations $F \subseteq Y$, we denote by $\Safe(F)$ the \emph{safety specification} and we define it as follows:
\begin{equation}
	\nonumber
	\Safe(F) := \{y_0 y_1 y_2 \cdots \in \Gamma_S \;\vert\; \forall k \geq 0 \; (y_k \in F)\}.
\end{equation}
The safety objective requires that the output of $S$ always remains within subset $F$.
Using LTL, such a safety specification is encoded as the formula $\Box F$. 
Similarly, for a target set of observations $T \subseteq Y$, we denote by $\Reach(T)$ the \emph{reachability specification} and we define it as follows:
\begin{equation}
	\nonumber
	\Reach(T) := \{y_0 y_1 y_2 \cdots \in \Gamma_S \;\vert\; \exists k \geq 0 \; (y_k \in T)\}.
\end{equation}
The reachability objective requires that the output of $S$ visits, at least once, some elements in $T$.
Such a reachability specification is encoded as the LTL formula $\Diamond T$.

Specifications like \emph{infinitely often} ($\Box \Diamond G$) and \emph{eventually forever} (a.k.a. \emph{persistence}) ($\Diamond \Box G$), for a set of observations $G \subseteq Y$, can be defined in a similar way.
It is also possible to extend the specifications to include timing constrains.
For example, $\Safe_{[a,b]}(F)$ and $\Reach_{[a,b]}(T)$ require that the output follows the specifications during the time steps $k \in \{a, a+1, \cdots, b\}$.
Such time-constrained specifications can be encoded in the form of metric temporal logic (MTL) formulae \cite{KRon_MTL}.

\begin{remark}
	For state-based systems, specifications are reduced automatically to sequences of states, since external and internal behaviors of systems coincide.
	In such a case, the satisfaction condition in Definition \ref{def_specification} should be checked against internal behaviors.
\end{remark}

Now, we \oldagain{formulate a general concept of}\newagain{introduce the} control problem considered in this article. \newagain{We then introduce controllers and their domains.}

\begin{problem}[Control Problem]
	\label{def_control_problem}
	Consider a system $S$ as defined in Definition \ref{def_system}.
	Let $\psi$ be a given specification on $S$ following Definition \ref{def_specification}.
	We denote by the tuple $(S,\psi)$ the control problem of finding a system $C$ such that $C \times S \models \psi$.
\end{problem}

\oldagain{Throughout the rest of the article, system $C$ is referred to as the \emph{controller} solving the control problem $(S,\psi)$. We define controllers formally next.}

\begin{definition}[Controller]
	\label{def_controller}
	Given a control problem $(S,\psi)$ as defined in Problem \ref{def_control_problem}, a controller solving the control problem is a feedback-composable system
	\begin{equation}
		\nonumber
		C := (X_C, X_{C,0}, U_C, \underset{C}{\longrightarrow}, Y_C, H_C),
	\end{equation}	
	where $U_C := Y$ and $Y_C := U$.
	All of $X_C$, $X_{C,0}$, $\underset{C}{\longrightarrow}$, and $H_C$ are constructed such that $C \times S \models \psi$.
\end{definition}

The domain of controller $C$ is the set of initial states of the controlled systems that can be controlled to solve the main control problem. We define it formally next.

\begin{definition}[Domain of Controller]
	\label{def_controller_domain}
	Consider a controller $C$ solving \old{a given control problem }$(S, \psi)$, as defined in Definition \ref{def_controller}.
	The domain of \old{the controller}\new{$C$} is denoted by $\mathcal{D}(C) \subseteq X_0$ and \old{it is }defined as follows:
	\begin{equation}
		\nonumber
		\mathcal{D}(C) := \{ x \in X_0 \;\vert\;  C \times S^{(\{x\})} \models \psi \}.
	\end{equation}
\end{definition}

\section{Output-Feedback Refinement Relations}
\label{sec_OFRR}
We first revise FRRs \cite{G.Reissig_etal_FRR_TAC}\old{ and show that they are not suitable for output-based systems}\new{ and then introduce OFRRs}.
\old{Then, we introduce OFRRs to relate two systems based on their outputs.}

\begin{definition}[FRR]
	\label{def_FRR}
	Consider two state-based systems $S_i := (X_i, X_{i,0}, U_i, \underset{i}{\longrightarrow}, X_i, id_{X_i})$, $i \in \{1,2\}$, and assume that $U_2 \subseteq U_1$.
	A strict relation $Q \subseteq X_1 \times X_2$ is an FRR from $S_1$ to $S_2$ if all of the followings hold for all $(x_1, x_2) \in Q$:
	\begin{enumerate}
		\item[(i)] $U_{S_2}(x_2) \subseteq U_{S_1}(x_1)$,
		\item[(ii)] $u \in U_{S_2}(x_2) \implies Q(\Post{x_1}{u}{S_1}) \subseteq \Post{x_2}{u}{S_2}$, and
		\item[(iii)] $x_1 \in X_{1,0} \implies x_2 \in X_{2,0}$.	
	\end{enumerate}
	When $Q$ is an FRR from $S_1$ to $S_2$, this is denoted by $S_1 \preccurlyeq_Q S_2$.
\end{definition}

FRRs are introduced to resolve common shortcomings in alternating (bi-)simulation relations (ASR) and their approximate versions.
As discussed in \cite{G.Reissig_etal_FRR_TAC}, using ASR results in controllers that require exact state information of concrete systems while only quantized state information is usually available.
Additionally, the refined controllers contain symbolic models of original systems as building blocks inside them, which makes the implementation much more complex.
On the other hand, controllers designed for systems related via FRRs require only quantized-state information.
They can be feedback-composed with original systems through static quantizers and they do not require the symbolic models as building blocks inside them.
Such features simplify refining and implementing the synthesized symbolic controllers.

Unfortunately, FRRs are only applicable to state-based systems.
Basically, a controller synthesized for the outputs of a symbolic model can not be refined to work with its original system.
This is because there is no mapping from the outputs of the original system to the outputs of its symbolic model.
Consequently, outputs of original systems received by the refined controllers cannot be matched to outputs of symbolic models used previously to synthesize the symbolic controllers.
We introduce OFRRs as extensions of FRRs so that one can construct symbolic models, synthesize symbolic controllers and refine them for output-based systems.

If $S_1$ and $S_2$ are \emph{output-based} systems, we use $S_1 \preccurlyeq_Q S_2$ to denote that $Q \subseteq X_1 \times X_2$ is an FRR from $S_{1,X_1}$ to $S_{2,X_2}$.

\begin{definition}[OFRR]
	\label{def_OFRR}
	Consider two output-based systems $S_i := (X_i, X_{i,0}, U_i, \underset{i}{\longrightarrow}, Y_i, H_i)$, $i \in \{1,2\}$, such that $U_2 \subseteq U_1$.
	Let $Q \subseteq X_1 \times X_2$ be an FRR such that $S_1 \preccurlyeq_Q S_2$.
	A relation $Z \subseteq Y_1 \times Y_2$ is an OFRR if all of the followings hold:
	\newagain{\begin{enumerate}
		\item[(i)]   For any $(y_1, y_2) \in Z,\; \bar{U}_{S_2}(y_2) \subseteq \bar{U}_{S_1}(y_1)$,
		\item[(ii)]  For any $(x_1, x_2) \in Q,\;  \exists (y_1, y_2) \in Z \;\text{s.t.}\; y_1 = H_1(x_1) \;\land\; y_2=H_2(x_2)$, and 
		\item[(iii)] For any $(y_1, y_2) \in Z,\; \exists (x_1, x_2) \in Q \;\text{s.t.}\; x_1 \in H_1^{-1}(y_1) \;\land\; x_2 \in H_2^{-1}(y_2)$.
	\end{enumerate}}
\end{definition}

Condition (i) ensures the admissibility of inputs of $S_2$ for $S_1$.
This is not restrictive for output-based systems representing control systems as we show later in Remark \ref{rmk_concrete_system_chars}.
Conditions (ii) and (iii) ensure that observed outputs correspond to evolving states that obey a valid FRR between the two systems.
For the sake of a simpler presentation, we slightly abuse the notation hereinafter and use $S_1 \preccurlyeq_Z S_2$ to indicate the existence of OFRR $Z$ from $S_1$ to $S_2$.

We provide a simple example to illustrate the importance of conditions (ii) and (iii).
\begin{example}
	\label{ex_OFRR_conditions}
	Consider system $S := (X, X, U, \longrightarrow, Y, H)$, where $X := \{x_1,x_2,x_3,x_4\}$, $U$ and $\longrightarrow$ are some sets, $Y := \{y_1,y_2,y_3,y_4\}$, and $H(x_i) := y_i$, for all $i \in \{1,2,3,4\}$.
	Also consider system $\sym{S} := (\sym{X}, \sym{X}, \sym{U}, \overset{\sym{}}{\longrightarrow}, \sym{Y}, \sym{H})$, where $\sym{X} := \{x_{q_1}, x_{q_2}, x_{q_3}\}$, $\overset{\sym{}}{\longrightarrow}$ and $\sym{U} \subseteq U$ are some sets, $\sym{Y} := \{y_{q_1}, y_{q_2}, y_{q_3}\}$, and $\sym{H}(x_{q_i}) := y_{q_i}$, for all $i \in \{1,2,3\}$.
	Consider a relation $Q := \{ (x_1, x_{q_1}), (x_2, x_{q_1}), (x_3, x_{q_2}), (x_4, x_{q_2}) \}$ and assume that the settings of $U$, $\longrightarrow$, $\sym{U}$, and $\overset{\sym{}}{\longrightarrow}$ ensure that $S \preccurlyeq_Q \sym{S}$.
	We inspect two relations $Z_1 \subseteq Y \times \sym{Y}$ and $Z_2 \subseteq Y \times \sym{Y}$ lacking, respectively, conditions (ii) and (iii) in Definition \ref{def_OFRR}:
	\begin{enumerate}
		\item A relation $Z_1 := \{ (y_1, y_{q_1}), (y_3, y_{q_2}), (y_4, y_{q_2}) \}$ violates condition (ii).
		Note that $(x_2, x_{q_1}) \in Q$ has no corresponding element in $Z_1$ that satisfies the condition.
		Consequently, if $S$ is at state $x_2$, its output $y_2 = H(x_2)$ can not be mapped to one of the outputs of $\sym{S}$.
		Condition (ii) ensures that, as system $S$ evolves, there always exit related observations in set $\sym{Y}$.
		
		\item A relation $Z_2 := \{ (y_1, y_{q_1}), (y_2, y_{q_1}), (y_3, y_{q_2}), (y_4, y_{q_2}), (y_3, y_{q_3}) \}$ violates condition (iii). 
		More precisely, $(y_3, y_{q_3}) \in Z$ has no corresponding element in $Q$ that satisfies condition (iii).
		Now, $\{y_{q_2} ,y_{q_3}\} = Z(y_3)$ makes it ambiguous to map the the output $y_3$ from $S$ to $\sym{S}$.
		Condition (iii) makes sure that outputs can be mapped unambiguously from $S$ to $\sym{S}$.
	\end{enumerate}
\end{example}

The following proposition provides sufficient conditions for the existence of OFRR.

\begin{proposition}
	\label{prop_OFRR_exists}
	Consider two systems $S_i := (X_i, X_{i,0}, U_i, \underset{i}{\longrightarrow}, Y_i, H_i)$, $i \in \{1,2\}$ having $U_2 \subseteq U_1$.
	Let $Q \subseteq X_1 \times X_2$ be an FRR such that $S_1 \preccurlyeq_Q S_2$\newagain{, $Y_2$ partitions $Y_1$, and
	\begin{equation}
		\label{eq_alignment_preservation_condition}
		y \in Y_2 \implies H_1(Q^{-1}(H_2^{-1}(y))) \equiv y.
	\end{equation}}
	Then, there exists a unique OFRR $Z \subseteq Y_1 \times Y_2$ corresponding to FRR $Q$ such that $S_1 \preccurlyeq_Z S_2$. 
\end{proposition}
\begin{proof}
	\label{proof_prop_OFRR_exists}
	Let $Q \subseteq X_1 \times X_2$ be an FRR.	
	We first prove by construction that $Z$ exists.
	Let $Z$ be as follows:
	\begin{align}
		\nonumber
		Z := \{ & (y_1, y_2) \;\vert\; \\ \nonumber 
		          &y_1 = H_1(x_1) \;\land\; y_2 = H_2(x_2) \;\text{for some}\; (x_1, x_2) \in Q \},
	\end{align}
	which satisfies conditions (i)-(iii) in Definition \ref{def_OFRR}.
	
	Now, we prove that $Z$ is unique.
	Consider two OFRRs $Z_1$ and $Z_2$ having the same underlying FRR $Q$.
	We show that they are equal.
	Consider any $(y_{1,1}, y_{1,2})\in Z_1$.
	We know from condition (iii) in the definition of OFRR $Z_1$ that there exists $(x_1, x_2) \in Q$ such that $x_1 \in H_1^{-1}(y_{1,1})$ and $x_2 \in H_2^{-1}(y_{1,2})$.
	We also know from condition (ii) in the definition of OFRR $Z_2$ that there exists $(y_1, y_2) \in Z_2$ such that $y_1 = H_1(x_1)$ and $y_2 = H_2(x_2)$.
	Clearly, $y_1 = y_{1,1}$ and $y_2 = y_{1,2}$ since the output maps are single-valued.
	This implies that $(y_{1,1}, y_{1,2})\in Z_2$ and, hence, $Z_1 \subseteq Z_2$.
	One can, similarly, show that $Z_2 \subseteq Z_1$ which proves that $Z_1 = Z_2$.
\end{proof}

The following corollary shows that OFRR and FRR coincide if the related systems are state-based.
\begin{corollary}
	\label{corr_OFRR_equals_FRR}
	Consider two systems $S_i:= (X_i, X_{i,0}, U_i, \underset{i}{\longrightarrow}, X_i, id_{X_i})$, $i \in \{1,2\}$, such that $U_2 \subseteq U_1$.
	Let $Q \subseteq X_1 \times X_2$ be an FRR such that $S_1 \preccurlyeq_Q S_2$.
	Then, $Z$ and $Q$ coincide (i.e., $Q = Z$). 
\end{corollary}
\begin{proof}
	\label{proof_corr_OFRR_equals_FRR}
	The proof is straightforward since we have $Q \subseteq Z$ and $Z \subseteq Q$ as a result of making $H_i = id_{X_i}$, $i \in \{1,2\}$ and using Definition \ref{def_OFRR}.
\end{proof}

\new{The following proposition shows that, when two systems are related via an OFRR $Z$ and as we observe one of the systems, we can always find corresponding outputs of the other system such that the successor outputs of both systems are in $Z$.}
\old{The following proposition shows that OFRRs ensure synchronized runs of underlying state-base systems.}
\new{Such a feature is used to prove the output-based behavioral inclusion from original systems to symbolic ones in Subsection \ref{ssec_controller_synthesis}.}
\begin{proposition}
	\label{prop_OFRR_evolution}
	Consider two systems $S_i := (X_i, X_{i,0}, U_i, \underset{i}{\longrightarrow}, Y_i, H_i)$, $i \in \{1,2\}$ having $U_2 \subseteq U_1$.
	Let $Z \subseteq Y_1 \times Y_2$ be an OFRR s.t. $S_1 \preccurlyeq_Z S_2$. 
	Then, for any $(y_1,y_2) \in Z$ we have:
	\begin{align}		
	\nonumber
	&\forall u \in \bar{U}_{S_2}(y_2) \; \forall y'_1 \in  \overline{\mathsf{Post}}_{u}^{S_1}(y_1) \; \exists y'_2 \in \overline{\mathsf{Post}}_{u}^{S_2}(y_2) \; \text{s.t.} \\ \nonumber
	&((y'_1, y'_2) \in Z).
	\end{align}
\end{proposition}
\begin{proof}
	\label{proof_prop_OFRR_evolution}
	Consider any $(y_1, y_2) \in Z$ and any $u \in \bar{U}_{S_2}(y_2)$.
	We know by condition (i) in Definition \ref{def_OFRR} that $u \in \bar{U}_{S_1}(y_1)$.	
	We also know from condition (iii) in Definition \ref{def_OFRR} that there exists $(x_1, x_2) \in Q$ such that $y_1 = H_1(x_1)$ and $y_2 = H_2(x_2)$.
	Now, consider any $x'_1 \in \Post{x_1}{u}{S_1}$.
	Also, consider the output of $x'_1$ which is $y'_1 = H_1(x'_1) \in H(\Post{x_1}{u}{S_1}) \subseteq \overline{\mathsf{Post}}_{u}^{S_1}(y_1)$.
	
	We know from Definition \ref{def_FRR} for $Q$ that $Q(x'_1) \subseteq \Post{x_2}{u}{S_2}$ which implies that there exists $x'_2 \in X_2$ such that $(x'_1, x'_2) \in Q$.
	From Definition \ref{def_OFRR} for $Z$, there exists $(y'_1, y^*) \in Z$ with $y^* = H_2(x'_2)$.
	What remains is to show that $y^* \in \overline{\mathsf{Post}}_{u}^{S_2}(y_2)$.
	By definition, we have $\overline{\mathsf{Post}}_{u}^{S_2}(y_2) = H_2(\Post{H_2^{-1}(y_2)}{u}{S_2})$. 
	We also know from Definition \ref{def_OFRR} that $x_2 \in H_2^{-1}(y_2)$ which implies that $x'_2 \in \Post{H_2^{-1}(y_2)}{u}{S_2}$.
	Note that $y^* = H_2(x'_2)$ implies that $y^* \in H_2(\Post{H_2^{-1}(y_2)}{u}{S_2}) = \overline{\mathsf{Post}}_{u}^{S_2}(y_2)$.
\end{proof}

\section{Output-Feedback Symbolic Control}
\label{sec_outfb_sym_control}

We first introduce control systems.
Then, we construct symbolic models of them and synthesize their symbolic controllers using OFRRs.

\subsection{Control systems}
\begin{definition}[Control System]
	\label{def_control_system}
	A control system is a tuple $\Sigma := (\mathcal{X}, \mathcal{U}, f, \mathcal{Y}, h)$, where 
	$\mathcal{X} \subseteq \R^n$ is the state set; 
	$\mathcal{U} \subseteq \R^m$ is an input set; 
	$f: \mathcal{X} \times \mathcal{U} \rightarrow \mathcal{X}$ is a continuous map satisfying the following Lipschitz assumption: for each compact set $\mathsf{X} \subseteq \mathcal{X}$, there exists a constant $L \in \R^+$ such that  
	\begin{equation}
		\nonumber		
		\Vert f(x_1,u) - f(x_2,u) \Vert \leq L \Vert x_1-x_2 \Vert,
	\end{equation}
	for all $x_1,x_2 \in \mathsf{X}$ and all $u \in \mathcal{U}$;
	$\mathcal{Y} = \R^q$ is the output set; and 
	$h: \mathcal{X} \rightarrow \mathcal{Y}$ is an output (a.k.a. observation) map.
\end{definition}

Let $\mathsf{U}$ be the set of all functions of time from $]a,b[ \subseteq \R$ to $\mathcal{U}$ with $a<0$ and $b>0$.
We define a trajectory of $\Sigma$ by \old{the}\new{a} locally absolutely continuous curve $\xi$ : $]a,b[ \rightarrow \mathcal{X}$  if there exists a $v \in \mathsf{U}$ that satisfies $\dot{\xi}(t) = f(\xi(t), v(t))$ at any $t \in ]a,b[$.
We redefine $\xi : [0,t] \rightarrow \mathcal{X}$ for trajectories over closed intervals with the understanding that there exists a trajectory $\xi' : ]a,b[ \rightarrow \mathcal{X}$ for which $\xi = \xi'|_{[0,t]}$ with $a<0$ and $b > t$.  
$\xi_{xv}(t)$ denotes the state reached at time $t$ under input $v$ and with the initial condition \new{$\xi_{xv}(0) = x$}. 
Such a state is uniquely determined since the assumptions on $f$ ensure the existence and uniqueness of its trajectories \cite{E.D.Sontag_MCT}.
System $\Sigma$ is said to be \textit{forward complete} if every trajectory is defined on an interval of the form $]a,\infty[$. 
Here, we consider forward complete control systems.
\old{We also define $\zeta : [0,t] \to \mathcal{Y}$ for output trajectories of $\Sigma$ such that, at any time $t \in [0,t]$, the output of $\Sigma$ satisfies $\zeta(t) = h(\xi_{xv}(t))$.}\new{We also define $\zeta : [0,t] \to \mathcal{Y}$ as an output trajectory of $\Sigma$ if there exists a trajectory $\xi_{xv}$ over $[0,t]$ such that at any time $\tilde{t} \in [0,t]$ we have that $\zeta(\tilde{t}) = h(\xi_{xv}(\tilde{t}))$.}

\subsection{Control Systems as Systems}		
Let $\Sigma$ be a control system as defined in Definition \ref{def_control_system}. 
The sampled version of $\Sigma$ (a.k.a. concrete system) is a system
\begin{align}
	\label{eq_concrete_output_system}
	S_\tau(\Sigma) := (X_\tau, X_\tau, U_\tau, \underset{\tau}{\longrightarrow}, Y_\tau, H_\tau),
\end{align}
that encapsulates the information contained in $\Sigma$ at sampling times $k\tau$, for all $k \in \N$, where $X_\tau \subseteq \mathcal{X}$, $U_\tau$ is the set of piece-wise constant curves of length $\tau$ defined as follows:
\new{\begin{align}
	\nonumber
	U_\tau 	&:= \{v_{\tau}: [0,\tau[ \to \mathcal{U} \;\vert\; \forall t\in[0,\tau[\;(v_{\tau}(t) = v_{\tau}(0))\},
\end{align}}
\new{$Y_\tau := \{y_\tau \in \mathcal{Y} \;\vert\; \exists x_\tau \in X_\tau \; (y_\tau = h(x_\tau)) \}$}, $H_\tau := h$, and a transition $(x_\tau, v_\tau,x'_\tau) \in \underset{\tau}{\longrightarrow} $ iff there exists a trajectory $\xi : [0, \tau] \rightarrow \mathcal{X}$ in $\Sigma$ such that $\xi_{x_\tau v_\tau}(\tau) = x'_\tau$. 
We sometimes use $S_\tau$ to refer to the sampled-data system $S_\tau(\Sigma)$.

\begin{remark}
	\label{rmk_concrete_system_chars}
	System $S_\tau$ is deterministic since any trajectory of\old{ the forward complete control system} $\Sigma$ is uniquely determined. 
	Sets $X_\tau$ and $U_\tau$ are uncountable, and hence, $S_\tau$ is not symbolic.
	Since all trajectories of $\Sigma$ are defined for all inputs and all states, we have $U_{S_\tau}(x_\tau) = U_\tau$, for all $x_\tau \in X_\tau$, and $\bar{U}_{S_\tau}(y) = U_\tau$, for all $y \in Y_\tau$.
\end{remark}

System $S_\tau$ is an output-based system. 
Any system feedback-composed with (or, serially composed after) $S_\tau$ has no access to its states, but rather to its outputs.
Throughout this article, we also consider a state-based version of $S_\tau$ (denoted by $S_{\tau, X}(\Sigma)$) and defined as follows:
\begin{align}
	\label{eq_state_based_system}
	S_{\tau, X}(\Sigma) := (X_\tau, X_\tau, U_\tau, \underset{\tau}{\longrightarrow}, X_\tau, id_{X_{\tau}}).
\end{align}

\subsection{Symbolic Models of Control Systems}		
We utilize OFRRs (and their underlying FRRs) to construct symbolic models that approximate $S_\tau$.
\old{Underlying FRRs of OFRRs ensure synchronized state-based evolution between concrete systems and their symbolic models.
OFRRs, on the other hand, ensure that outputs of symbolic models can be identified from outputs of their concrete systems.}
Given a control system $\Sigma$, let $S_\tau$ be its sampled-data representation, as defined in \eqref{eq_concrete_output_system}.
A symbolic model of $S_{\tau}$ is a system:
\begin{equation}
	\label{eq_abstract_output_system}
	\sym{S} :=(\sym{X}, \sym{X},\sym{U},\underset{q}{\longrightarrow}, \sym{Y}, \sym{H}),
\end{equation}
where 
$\sym{X} := X_\tau/\bar{Q}$, 
$\bar{Q}$ is a finite equivalence relation on $X_\tau$,
$\sym{U}$ is a finite subset of $U_{\tau}$,
$(\sym{x}, \sym{u}, \sym{x}') \in \underset{q}{\longrightarrow}$ if there exist $x \in \sym{x}$ and $x' \in \sym{x}'$ such that $(x,\sym{u},x') \in \underset{\tau}{\longrightarrow}$, 
$\sym{Y} := H_\tau(X_\tau) / \bar{Z}$,
where $\bar{Z}$ is a finite equivalence relation on $Y_\tau$, 
\old{$\sym{H}(\sym{x}) := H_\tau(x_q)$,} 
\new{$\sym{H}(\sym{x}) := \{ \sym{y} \in \sym{Y} \;\vert\; \sym{y} \cap H_{\tau}(\sym{x}) \neq \emptyset \}$,}
and \oldagain{the following }condition \newagain{\eqref{eq_alignment_preservation_condition}} holds \newagain{for $S_1 := S_{\tau}$ and $S_2 := \sym{S}$}.
\old{\begin{align}
	\label{eq_alignment_preservation_condition}
	 \{\sym{y} \in \sym{Y}\;\vert\;& \exists \sym{x} \in \sym{X} \text{ s.t. }\\
	 \nonumber
	 &  (\sym{x} = H^{-1}_{\tau}(\sym{y}) \;\land\; \sym{y} = H_{\tau}(\sym{x})) \} = \sym{Y}.
\end{align}}
\oldagainnostrike{\begin{align}
	\sym{y} \in \sym{Y} \implies H_\tau(H_q^{-1}(\sym{y})) \equiv \sym{y}.
\end{align}}

\old{
Condition \eqref{eq_alignment_preservation_condition} ensures the alignment between equivalence classes of $\bar{Z}$ (e.g., $\sym{y} \in \sym{Y}$) and $\bar{Q}$.
}

\old{
\begin{remark}
	\label{rmk_using_covers_is_doable}
	For a simpler presentation of the rest of this section (Theorems \ref{thrm_OFRR_exists_for_all_Q_quotent} and \ref{thrm_external_behavioral_inclusion}) and the implementations in Section \ref{sec_examples}, we choose to construct $\sym{X}$ and $\sym{Y}$ using equivalence relations.
	Nevertheless, the presented results can be extended such that $\sym{X}$ and $\sym{Y}$ are covers of $X_\tau$ and $Y_\tau$, respectively.
	Note that the construction of symbolic models using FRR in \cite{G.Reissig_etal_FRR_TAC} originally defines $\sym{X}$ as a cover of $X_\tau$.
\end{remark}
}

Starting with a given equivalence relation $\bar{Z}$ on $Y_\tau$, one can construct the underlying equivalence relation $\bar{Q}$ on $X_\tau$ using the following relation condition for any $(x_a,x_b) \in \bar{Q}$:
\old{\begin{align}
	\label{cond_eq_rel}
	x_a \sim x_b \; \iff \; &\exists (y_a,y_b) \in Y_{\tau} \times Y_{\tau} \text{ s.t. } \\
	\nonumber
	& (x_a \in H^{-1}_{\tau}(y_a) \;\land\; x_b \in H^{-1}_{\tau}(y_b) \;\land\; \\
	\nonumber
	& [y_a] = [y_b]),
\end{align}}
\begin{equation}
	\label{cond_eq_rel}
	x_a \sim x_b \; \iff \; (H_{\tau}(x_a),H_{\tau}(x_b)) \in \bar{Z},
\end{equation}
which ensures that condition \eqref{eq_alignment_preservation_condition} is satisfied.
\old{One can readily verify that condition \eqref{cond_eq_rel} implies that $\bar{Q}$ is indeed an equivalence relation as it satisfies reflexivity, symmetry, and transitivity conditions of equivalence relations.}
The following theorem shows that the above introduced construction of $\sym{S}$ implies the existence of some OFRR $Z$ such that $S_\tau \preccurlyeq_Z \sym{S}$.

\begin{theorem}
	\label{thrm_OFRR_exists_for_all_Q_quotent}
	Let $S_\tau$ be defined as in \eqref{eq_concrete_output_system}.
	Also, let $\sym{S}$ be defined as in \eqref{eq_abstract_output_system} for some equivalence relations $\bar{Q}$ on $X_\tau$ and $\bar{Z}$ on $H_\tau(X_\tau)$ \oldagain{satisfying condition }\oldagainnostrike{\eqref{eq_alignment_preservation_condition}}. 
	Then, 
	\begin{equation}
		\nonumber
		Z := \{ (y,[y]) \in Y_\tau \times \sym{Y} \;\vert\; y \in H_\tau(X_\tau) \},
	\end{equation}
	is an OFRR such that $S_\tau \preccurlyeq_Z \sym{S}$ and
	\begin{equation}
		\nonumber
		Q := \{ (x,[x]) \in X_\tau \times \sym{X} \;\vert\; x \in X_\tau \},
	\end{equation}	
	is its underlying FRR.
\end{theorem}
\begin{proof}
	First, we show that $Q$ is an FRR.
	Clearly, conditions (i) and (iii) in Definition \ref{def_FRR} hold since $S_\tau$ represents a control system.
	See Remark \ref{rmk_concrete_system_chars} for more details.
	We show that condition (ii) holds.
	Consider any $(x,[x]) \in Q$ and any input $\sym{u} \in \sym{U}([x])$.
	Also consider any successor state $x' \in \Post{x}{\sym{u}}{S_\tau}$.
	Remark that $x \in [x]$ and $x' \in [x']$ since $Q$ is an equivalence relation.
	Now, from the definition of $\sym{S}$ in \eqref{eq_abstract_output_system}, we know that there exits a corresponding transition $([x],\sym{u},[x'])$ in $\underset{q}{\longrightarrow}$.
	Since, $[x'] \in Q(x)$, by the definition of $Q$, we have that $[x'] \in \Post{[x]}{\sym{u}}{\sym{S}}$.
	Consequently,  $Q$ is an FRR from $S_{\tau, X_{\tau}}$ to ${\sym{S}}_{,\sym{X}}$.
	
	Now, we show that $Z$ is an OFRR.
	Again, condition (i) in Definition \ref{def_OFRR} holds since $S_\tau$ represents a control system.
	
	We show that condition (ii) in Definition \ref{def_OFRR} holds.
	Consider any $(x,[x]) \in Q$.
	Since $x \in X_\tau$, there exists one observation $y := H_\tau(x)$.
	Note that $[x] \in \sym{X}$.
	Now, by the definition of $\sym{Y}$ in \eqref{eq_abstract_output_system}, we know there exists $[y] \in \sym{Y}$ such that $[y] = \sym{H}([x])$.
	Finally, by the definition of $Z$, which is based on the equivalence relation $\bar{Z}$, we have that $(y,[y]) \in Z$, and this, consequently, satisfies condition (ii) in Definition \ref{def_OFRR}.
	
	We show that condition (iii) in Definition \ref{def_OFRR} holds.
	Consider any $(y,[y]) \in Z$.
	Note that $y \in H_\tau(X_\tau)$ (i.e., inside the the image of $X_\tau$ using $H_\tau$).
	From the definition of system $S_\tau$ in \eqref{eq_concrete_output_system}, we know that there exits $x \in X_\tau$ such that $x = H^{-1}_{\tau}(y)$.
	Also, we know from condition \eqref{eq_alignment_preservation_condition} and the definition of $\sym{X}$ that there exits $[x] \in \sym{X}$ such that $[x] = H^{-1}_{\tau}([y])$.
	Finally, by the definition of $Q$, which is based on the equivalence relation $\bar{Q}$, we conclude that $(x,[x]) \in Q$, and this, consequently, satisfies condition (iii) in Definition \ref{def_OFRR}.

	Now, recall Proposition \ref{prop_OFRR_exists}, and set $S_1 := S_{\tau}$ and $S_2 := \sym{S}$.
	Hence, we have that $S_\tau \preccurlyeq_Z \sym{S}$.
\end{proof}

\old{
\begin{remark}
	\label{rmk_Z_Q_coincide_in_state_based}
	If the following state-based version of the symbolic model is considered:
	\begin{align}
		\nonumber
		{\sym{S}}_{,\sym{X}} :=(\sym{X}, \sym{X},\sym{U},\underset{q}{\longrightarrow}, \sym{X}, id_{\sym{X}}),
	\end{align}
	the result from Corollary \ref{corr_OFRR_equals_FRR} applies and the OFRR $Z$ coincides with the FRR $Q$.
\end{remark}
}

\subsection{Synthesis and Refinement of Symbolic Controllers}
\label{ssec_controller_synthesis}
\oldagain{Consider a concrete system $S_{\tau}$ and let $\sym{S}$ be its \emph{symbolic model}, as introduced in }\oldagainnostrike{\eqref{eq_abstract_output_system}}\oldagain{, such that $\sym{S}$ is related to $S_{\tau}$ via some OFRR $Z$.}
Let $\sym{\psi}$ be a given output-based specification on $\sym{S}$\newagain{ as introduced in \eqref{eq_abstract_output_system}}.
$\psi_\tau$ is the corresponding concrete specification that should be enforced on $S_{\tau}$ and it is interpreted as follows:
\begin{align}
	\label{eq_from_concrete_specs_to_abstract_specs}
	\psi_\tau := \{ \bar{s} \in \Gamma_{S_\tau} \;\vert\; 
	& \exists s \in \sym{\psi} \; \forall i \in \{0,1,\cdots,\vert s \vert-1\} \; \\ \nonumber 
	& ( s_i = Z(\bar{s}_i)) \}.
\end{align}
Here, $S_{\tau}$ and $\psi_\tau$ represent together a \emph{concrete control problem} $(S_{\tau},\psi_\tau)$, whereas $(\sym{S},\sym{\psi})$ represents an \emph{abstract control problem}.
To algorithmically design controllers solving $(S_{\tau},\psi_\tau)$, we utilize $(\sym{S},\sym{\psi})$ to automatically synthesize a \emph{symbolic controller} $\sym{C}$ that can be refined to solve $(S_{\tau},\psi_\tau)$.
Later in Section\old{s \ref{meth_game_based}, \ref{mthd_concrete_domain_observers} and} \ref{meth_detector}, we propose \old{some methodologies}\new{a methodology} for synthesizing $\sym{C}$, which is then refined with a suitable interface to a controller $C_\tau$ that solves the concrete control problem $(S_\tau,\psi_\tau)$.

Now, we \old{explain the refinement process.
We }show that OFRRs preserve the behavioral inclusion from concrete systems to symbolic models.
\old{More specifically, we show in the next result that the internal (resp., external) behavior of a refined concrete closed-loop $C_\tau \times S_\tau$ is included in the internal (resp., external) behaviors of its corresponding symbolic closed-loop $\sym{C} \times \sym{S}$.}

\begin{theorem}
	\label{thrm_external_behavioral_inclusion}
	Consider systems $S_\tau$ and $\sym{S}$ as introduced in \eqref{eq_concrete_output_system} and \eqref{eq_abstract_output_system}, respectively, where $Z$ is an OFRR and $S_\tau \preccurlyeq_Z \sym{S}$.
	Let $\sym{C}$ be a controller that solves $(\sym{S},\sym{\psi})$.
	Then,
	\new{\begin{enumerate}
		\item[(i)] $(\sym{C} \circ Z)$ is feedback-composable with $S_\tau$;
		\item[(ii)] $Z(B_{int}((\sym{C} \circ Z) \times S_{\tau})) \subseteq B_{int}(\sym{C} \times \sym{S})$; and
		\item[(iii)] $Z(B_{ext}((\sym{C} \circ Z) \times S_{\tau})) \subseteq B_{ext}(\sym{C} \times \sym{S})$;
	\end{enumerate}}
\end{theorem}
\begin{proof}
	\label{proof_thrm_external_behavioral_inclusion}
	
	Proof of (i):
	Let system $C_\tau$ be of the form
	\begin{equation}
		\nonumber
		C_\tau := \sym{C} \circ Z := (X_C, X_{C,0}, U_C, \underset{C}{\longrightarrow}, Y_C, H_C),
	\end{equation}
	for some sets $X_C$, $X_{C,0}$, $U_C$, $\underset{C}{\longrightarrow}$, and $Y_C$, and a map $H_C$.
	Now, based on the given assumptions and \new{\cite[Definition III.2]{G.Reissig_etal_FRR_TAC}}, we have that $Y_C \subseteq \sym{U}$ and $Y_\tau \subseteq U_C$.
	Since $\sym{U} \subseteq U_\tau$, we know that $Y_C \subseteq U_\tau$.
	From Definition \ref{def_FRR} and since $\sym{C}$ is feedback-composable with $\sym{S}$, we get 
	\begin{align}
		\nonumber
		&\sym{y} = \sym{H}(\sym{x}) \;\land\; \sym{u} = H_{\sym{C}}(x_{\sym{C}}) \;\land\; \Post{x_{\sym{C}}}{\sym{y}}{\sym{C}} = \emptyset \\
		\nonumber
		&\implies \Post{\sym{x}}{\sym{u}}{\sym{S}} = \emptyset.
	\end{align}
	From condition (i) in Definition \ref{def_OFRR} and considering $Z$ as a serially composed static map with $\sym{C}$, we get
	\begin{align}
	\nonumber
		&y_\tau = H_\tau(x_\tau) \;\land\; \sym{u} = H_{C}(x_{C}) \;\land\; \Post{x_{C}}{y_\tau}{C_\tau} = \emptyset \\
		\nonumber
		&\implies \Post{x_\tau}{\sym{u}}{S_\tau} = \emptyset,
	\end{align}	
	which completes the proof of (i).
	
	Proof of (ii): 
	The results in \cite[Theorem V.4]{G.Reissig_etal_FRR_TAC} are directly applicable here since $S_{\tau, X_\tau}$ and $S_{q, X_q}$ are \old{simple systems}\new{state-based systems} that are related via an FRR.
	This completes the proof of (ii).
	
	To proof (iii), 
	consider any external run $r_{C_\tau \times S_{\tau}, ext} \in B_{ext}(C_\tau \times S_{\tau})$ defined as:
	\begin{equation}
		\nonumber
		r_{C_\tau \times S_{\tau}, ext} := 
			(\sym{u}^0, y_\tau^0)0
			(\sym{u}^1, y_\tau^1)0 \cdots 
			(\sym{u}^i, y_\tau^i)0 \cdots,
	\end{equation}
	where $i \in \N$.
	According to Proposition \ref{prop_FB_runs}, there exist two external runs:	
	\begin{equation}
		\nonumber
		r_{C_\tau, ext}  := 
			\sym{u}^0 y_\tau^0 
			\sym{u}^1 y_\tau^1 \cdots 
			\sym{u}^i y_\tau^i \cdots, \text{ and }
	\end{equation}		
	\begin{equation}	
		\label{eq_proof_thrm_bi_external_runs}	
		r_{S_{\tau}, ext} := 
			y_\tau^0 \sym{u}^0 
			y_\tau^1 \sym{u}^1 \cdots 
			y_\tau^i \sym{u}^i \cdots, 
	\end{equation}
	where $i \in \N$.
	
	Notice how the output sets $Y_{\tau}$ and $\sym{Y}$ are constructed in \eqref{eq_concrete_output_system} and \eqref{eq_abstract_output_system}, respectively.
	Both of them use map $H_{\tau}$ to project the state set $X_{\tau}$.
	Then, one can easily show that $Z$ is a strict relation.
	\new{Now, using the given relation $Z$ and }\old{Consequently, }for any $y_\tau^i$, $i \in \N$, \new{we know that }there exists a corresponding $\sym{y}^i \in \sym{Y}$ such that $(y_\tau^i,\sym{y}^i) \in Z$.	\new{This allows us to apply $Z$ on the concrete output elements of each of the runs in \ref{eq_proof_thrm_bi_external_runs}.}
	
	Now, by applying Proposition \ref{prop_OFRR_evolution} inductively to \eqref{eq_proof_thrm_bi_external_runs} starting with $(y_\tau^0, \sym{y}^0) \in Z$, we conclude that the following external run $r_{\sym{S}, ext} \in B_{ext}(\sym{S})$ exits:
	\begin{equation}		
		\nonumber
		r_{\sym{S}, ext} := 
			\sym{y}^0 \sym{u}^0 
			\sym{y}^1 \sym{u}^1 \cdots 
			\sym{y}^i \sym{u}^i \cdots.
	\end{equation}
	
	Also, since map $Z$ is strict, and it interfaces the input to $\sym{C}$, one can assume that run $r_{C_\tau, ext} \in B_{ext}(C_\tau)$ is synchronized with an run $r_{\sym{C}, ext} \in B_{ext}(\sym{C})$ given by:
	\begin{equation}
		\nonumber
		r_{\sym{C}, ext}  := 
		\sym{u}^0 \sym{y}^0 
		\sym{u}^1 \sym{y}^1 \cdots 
		\sym{u}^i \sym{y}^i \cdots.
	\end{equation}	
	
	Again, according to Proposition \ref{prop_FB_runs}, the two runs $r_{\sym{C}, ext}$ and $r_{\sym{S}, ext}$ imply the existence of the external run of the feedback-composed system $\sym{C} \times \sym{S}$:	
	\begin{equation}
		\nonumber
		r_{\sym{C} \times \sym{S}, ext} := 
			(\sym{u}^0, \sym{y}^0)0
			(\sym{u}^1, \sym{y}^1)0 \cdots 
			(\sym{u}^i, \sym{y}^i)0 \cdots,
	\end{equation}
	where $i \in \N$, which proves that $r_{(\sym{C} \circ Z) \times S_{\tau}, ext} \in B_{ext}(\sym{C} \times \sym{S})$, and completes the proof of (iii).
	
\end{proof}

The following corollary shows that internal behavioral inclusion from a concrete closed-loop to a symbolic closed-loop implies an external behavioral inclusion.
\begin{corollary}
	\label{corr_internal_behavior_implies_extrnal_behavior}
	Let $S_\tau$ and $\sym{S}$ be as introduced in \eqref{eq_concrete_output_system} and \eqref{eq_abstract_output_system}, respectively, where $Z$ is an OFRR and $S_\tau \preccurlyeq_Z \sym{S}$.
	Then,
	\begin{align}
		\nonumber
		B_{int}((\sym{C} &\circ Z) \times S_{\tau}) \subseteq B_{int}(\sym{C} \times \sym{S}) \implies \\ \nonumber
		& B_{ext}((\sym{C} \circ Z) \times S_{\tau}) \subseteq B_{ext}(\sym{C} \times \sym{S}).
	\end{align}
\end{corollary}
\begin{proof}
	\label{proof_internal_behavior_implies_extrnal_behavior}
	The proof is similar to that of part (iii) in Theorem \ref{thrm_external_behavioral_inclusion} by mapping the internal sequences to external sequences.
\end{proof}

\begin{remark}
	\label{rmk_direct_refinement}
	Given two systems $S_\tau$ and $\sym{S}$ such that $S_\tau \preccurlyeq_Z \sym{S}$, for some OFRR $Z$, a controller $\sym{C}$ that solves the abstract control problem $(\sym{S}, \sym{\psi})$ can be refined to solve the concrete control problem $(S_\tau, \psi_\tau)$ using $Z$ as a static map.
\end{remark}

\newagain{
\begin{remark}
	\label{rmk_general_framework}
    Theorem \ref{thrm_external_behavioral_inclusion} and Corollary \ref{corr_internal_behavior_implies_extrnal_behavior} provide general results for output-feedback symbolic control. 
    They can be applied to any methodology that can synthesize controllers (cf. Definition \ref{def_controller}) for the outputs of symbolic models (cf. the definition in \eqref{eq_abstract_output_system}) to enforce output-based specifications (cf. Definition \ref{def_specification}).
\end{remark}

The next three sections provide example methodologies that realize the introduced framework.
}


\begin{figure}
	\centering
	\includegraphics[width=0.5\textwidth]{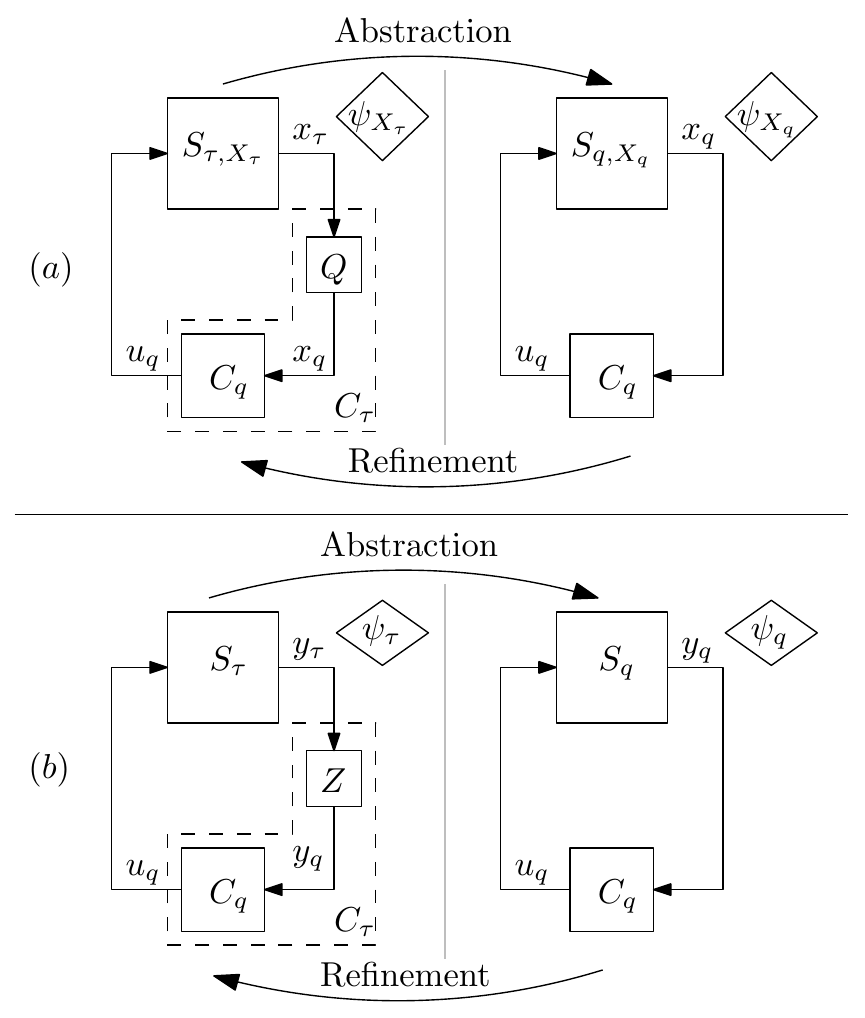}
	\caption{
		Synthesis and refinement of (a) state-based and (b) output-based systems.
		Systems are represented by rectangles, specifications by diamonds and interconnections by arrows.
	}	
	\label{fig_synthesis_refinement_map}
\end{figure}

Figure \ref{fig_synthesis_refinement_map} provides an illustration for the synthesis and refinement of symbolic controllers for state-based and output-based systems.
For state-based systems, the refined controller is the symbolic controller serially composed after the map $Q$, i.e. $C_\tau = Q \circ \sym{C}$ \cite{G.Reissig_etal_FRR_TAC}.
For output-based systems, the refined controller is the symbolic controller serially composed after the map $Z$, i.e. $C_\tau = Z \circ \sym{C}$.

The presented results serve as a generalized framework that formulates the synthesis and refinement of symbolic controllers for output-based systems.
What remains is to provide specific implementations that show how symbolic controllers are synthesized and refined.
In the following sections, we present three different methodologies to serve this purpose.

\section{Methodology 1: Games of Imperfect Information}
\label{meth_game_based}
\newcommand{\PlayerOne}{{\tt Player1}}
\newcommand{\PlayerTwo}{{\tt Player2}}

\begin{figure}
	\centering
	\includegraphics[width=0.6\textwidth]{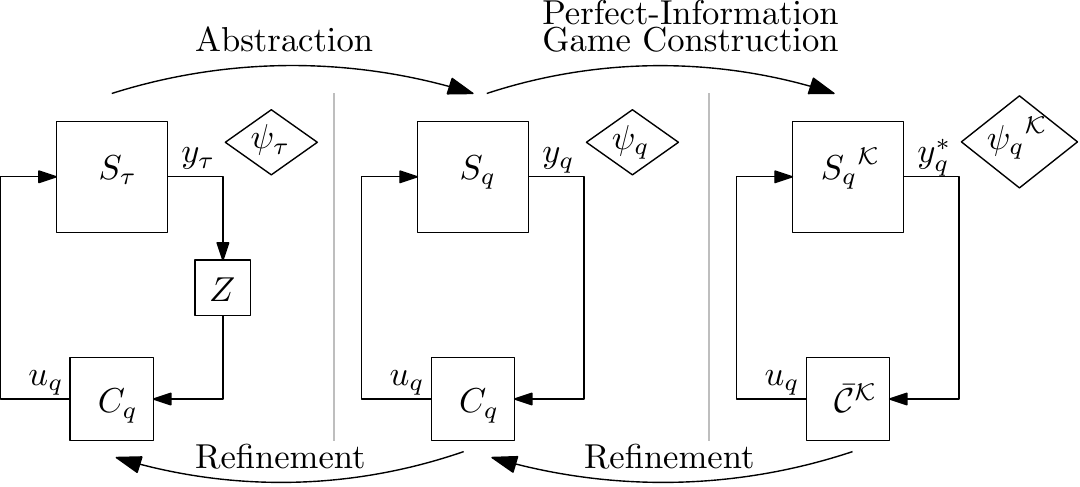}
	\caption{
		Output-based symbolic control using two-player games with imperfect information.
	}	
	\label{fig_games_method}
\end{figure}

Two-player games on graphs arise in many computer science problems \cite{ROTH1978_twoperson_games}.
We utilize the results in \cite{REIF1984274,Chatterjee2006_omegareg_games,Berwanger2008} and construct perfect-information (a.k.a. knowledge-based) games from output-based symbolic models.
Then, we solve the abstract control problem (or \emph{the game}) as presented in \cite{Berwanger2008}.
We then refine the synthesized controller in two steps: 
1) the symbolic controller synthesized for the game structure is refined to work with the symbolic model, and 
2) Theorem \ref{thrm_external_behavioral_inclusion} is used to refine the controller once again for the concrete system.
Figure \ref{fig_games_method} provides a high-level overview of this methodology.

\subsection{Output-based Symbolic Control using Two-player Games}
We assume having a symbolic model $\sym{S}$, as defined in \eqref{eq_abstract_output_system}, related via an OFRR $Z$ to a sampled output-based system $S_{\tau}$, as defined in \eqref{eq_concrete_output_system}.
The following assumptions are required \cite{Chatterjee2006_omegareg_games}:
\begin{enumerate}
	\item the abstract system $\sym{S}$ is total; and
	\item the set $\{ \sym{H}^{-1}(\sym{y}) \vert \sym{y} \in \sym{Y} \}$ partitions $\sym{X}$.
\end{enumerate}

The first assumption is not restrictive since all inputs are admissible to all states in control systems (see Remark \ref{rmk_concrete_system_chars}).
The second assumption is already satisfied as we consider quotient systems, based on the definition of the symbolic model in \eqref{eq_abstract_output_system} and the result from Theorem \ref{thrm_OFRR_exists_for_all_Q_quotent}.

The symbolic model $\sym{S}$ is seen as a \emph{game structure} of two players played in rounds.
The symbolic controller $C_q$ is named \PlayerOne{} and, at each game round, it selects an input $\sym{u} \in \sym{U}$ for the game structure $\sym{S}$.
A hypothetical player \PlayerTwo{}, or simply the symbolic model itself, responds by resolving the nondeterminism and selects a successor $\sym{x}'$ for the state $\sym{x}$ using the supplied input $\sym{u}$ such that $(\sym{x},\sym{u},\sym{x}') \in \underset{q}{\longrightarrow}$.

$\sym{S}$ is considered as a game structure of imperfect information since \PlayerOne{} has no access to the states of the game.
During the game play, only observations of the game structure are available to \PlayerOne{}.
Given an internal run (a.k.a. a \emph{play}) $r_{\sym{S}, int}$, we construct a corresponding external run $\sym{obs}$ as the unique sequence of observations:
\begin{align}
	\nonumber
	\sym{obs} :&= r_{\sym{S}, ext} \\ \nonumber
			   &= \sym{H}(\pi_{\sym{X}}(r_{\sym{S}, int})) = y_{q,0}y_{q,1} \cdots y_{q,n-1}y_{q,n} \cdots.
\end{align}

The  \emph{knowledge} associated with the prefix $\sym{obs}(n) := y_{q,0} y_{q,1} \cdots y_{q, n-1}y_{q, n}$ is given by the set:
\begin{align}
	\nonumber
	\mathcal{K}&(\sym{obs}(n)) := \{\Last{r_{\sym{S}, int}(n)} \;\vert\; \\ \nonumber
	& r_{\sym{S}, int}(n) \in \mathsf{PREFS}_{int}(\sym{S}) \;\land\; H(r_{\sym{S}, int}(n)) =  \sym{obs}(n)\},
\end{align}
which represents the set of possible underlying states expected at the end of the monitored observation sequence.
Having an initial knowledge $s_0 := X_{q,0} \subseteq \sym{X}$, the knowledge $s_i := \mathcal{K}(\sym{obs}(i))$, at any step $i \in \N$, $i \ge 1$, can be constructed iteratively \cite[Lemma 2.1]{Chatterjee2006_omegareg_games} using the received observation and the input \cite{Berwanger2008}:
\begin{equation}
\nonumber
s_i := \Post{s_{i-1}}{u_{q,i}}{\sym{S}} \cap \sym{H}^{-1}(\Last{\sym{obs}(i)}),
\end{equation}
where $u_{q,i}$ is the input at time step $i$.

\begin{remark}
	Since \PlayerOne{} generates the inputs, it can construct the knowledge at every step by having $s_0$ and monitoring the observations of the game structure.
\end{remark}

\subsection{Controller Synthesis and Refinement}
Consider a concrete game $(S_\tau, \psi_\tau)$ and its corresponding abstract game $(\sym{S}, \sym{\psi})$, where $S_\tau \preccurlyeq_Z \sym{S}$ and the specification $\sym{\psi}$ is constructed from $\psi_\tau$ using $Z$ as a static map as introduced in \eqref{eq_from_concrete_specs_to_abstract_specs}.
The first goal is to synthesize a controller $\sym{C}$ that solves $(\sym{C} \times \sym{S})$.

A \emph{strategy} for \PlayerOne{} is a map $\mathcal{C}: \sym{Y}^* \to \sym{U}$ that accepts a sequence of observations and produces a control input.
$\mathcal{C}$ is said to be memoryless strategy (a.k.a. a static controller) if $\mathcal{C}(r\cdot\sym{y}) = \mathcal{C}(r'\cdot\sym{y})$ for all $r,r' \in \sym{Y}^*$.
A memoryless strategy $\mathcal{C}$ induces another strategy $\bar{\mathcal{C}}:\sym{Y}\to\sym{U}$ that works with the last element of the observation for which $\mathcal{C}(r) = \bar{\mathcal{C}}(\Last{r})$ for all $r \in \sym{Y}^*$.

Having a strategy $\mathcal{C}$, we denote by $\mathsf{Outcome}_{\sym{S}}(\mathcal{C})$ the set of all possible state sequences resulting from closing the loop between $\sym{S}$ and $\mathcal{C}$, and we define it as follows:
\begin{align}
	\nonumber
	&\mathsf{Outcome}_{\sym{S}}(\mathcal{C}) := \{ x_{q,0} x_{q,1} \cdots \;\vert\; x_{q,0} \in X_{q,0} \;\land\; \\ 
	\nonumber
	& (\forall i \geq 0, (x_{q,i}, u_{q, i}, x_{q,i+1}) \in \underset{q}{\longrightarrow} \;\land\; u_{q,i} = \mathcal{C}(\sym{obs}(i))) \}.
\end{align}
We say that game $(\sym{S}, \sym{\psi})$ is solvable when there exists a strategy $\mathcal{C}$ such that for all $r \in \mathsf{Outcome}_{\sym{S}}(\mathcal{C})$, we have $\sym{H}(r) \in \sym{\psi}$. 
The strategy is then called a winning strategy.

To check the existence of a winning strategy, we construct another game of perfect information \cite{REIF1984274,Berwanger2008}.
The knowledge-based perfect-information game structure is a system:
\begin{equation}
	\nonumber
	\sym{S}^\mathcal{K} := (X_\mathcal{K}, s_0, \sym{U}, \underset{\mathcal{K}}{\longrightarrow}, X_\mathcal{K}, id_{X_\mathcal{K}}),
\end{equation}
where $X_\mathcal{K} := 2^{\sym{X}}\backslash \emptyset$, and $(s_1, \sym{u}, s_2) \in \underset{\mathcal{K}}{\longrightarrow}$ iff there exists an observation $\sym{y} \in \sym{Y}$ such that:
\begin{equation}
	\label{eq_knowldge_construction}
	s_2 := \Post{s_1}{\sym{u}}{\sym{S}^\mathcal{K}} \cap \sym{H}^{-1}(\sym{y}).
\end{equation}

\begin{proposition}
	\label{prop_game_solving}
	\PlayerOne{} has a winning strategy in the game $\sym{S}$ starting at the initial set $X_{q,0}$ iff \PlayerOne{} has a winning strategy in $\sym{S}^{\mathcal{K}}$ starting at $X_{q,0}$.
\end{proposition}
\begin{proof}
	The proof is similar to that in \cite[Proposition 2.1]{REIF1984274} and \cite[Proposition 2.4]{Chatterjee2006_omegareg_games}.
\end{proof}

In \cite[Algorithm 1]{Berwanger2008}, the game of imperfect information is solved using an antichain-based technique.
The technique is implemented in a tool named \ALPAGA{} \cite{ALPAGA}.
Using the tool, one can possibly synthesize a winning memoryless strategy $\bar{\mathcal{C}}^\mathcal{K}$ for the game $(\sym{S}^\mathcal{K}, \sym{\psi}^\mathcal{K})$, where $\sym{\psi}^\mathcal{K}$ is an extended version of $\sym{\psi}$ constructed by the same tool.
The memoryless strategy is refined to work with $\sym{S}$ by embedding it inside the symbolic controller $\sym{C}$:
\begin{equation}
\label{eq_controller_game}
\sym{C} := (X_{\sym{C}}, X_{\sym{C},0}, U_{\sym{C}}, \underset{\sym{C}}{\longrightarrow}, Y_{\sym{C}}, H_{\sym{C}}),
\end{equation}
where 
\begin{itemize}
	\item $X_{\sym{C}} := \sym{U} \times 2^{\sym{X}}$;
	\item $X_{\sym{C},0} := \{u_{q,0}\} \times \sym{X}$, where $u_{q,0} \in \sym{U}$;
	\item $U_{\sym{C}} := \sym{Y}$;
	\item $\underset{\sym{C}}{\longrightarrow} := \{((\sym{u}, x_{\sym{C}}),\sym{y},(\sym{u}', x_{\sym{C}}')) \vert x_{\sym{C}}' = \Post{x_{\sym{C}}}{\sym{u}}{\sym{S}} \cap \sym{H}^{-1}(\sym{y}) \;\land\; \sym{u}'=\bar{\mathcal{C}}^\mathcal{K}(x_{\sym{C}}') \}$;
	\item $Y_{\sym{C}} := \sym{U}$; and
	\item $H_{\sym{C}} := \pi_{\sym{U}}$.
\end{itemize}

\begin{remark}
	\label{rmk_game_controller_not_static}
	The strategy $\bar{\mathcal{C}}^\mathcal{K}$ synthesized via the knowledge-based game is static. 
	The refined game controller $\sym{C}$ contains the symbolic model $\sym{S}$ as a building block inside it, in order to compute the knowledge and, hence, it is not static anymore.
\end{remark}

The following theorem shows how the controller is refined and concludes this section.
\begin{theorem}
	Let $(S_\tau, \psi_\tau)$ be a concrete game and $(\sym{S}, \sym{\psi})$ be an abstract game, where $S_\tau \preccurlyeq_Z \sym{S}$ and $\sym{\psi}$ is a specification constructed from $\psi_\tau$ using $Z$ as a static map.
	If a controller $\sym{C}$, as defined in \eqref{eq_controller_game}, solves the game $(\sym{S}, \sym{\psi})$ then $(\sym{C} \circ Z)$ solves the game $(S_\tau, \psi_\tau)$.
\end{theorem}
\begin{proof}
	The proof follows directly from Proposition \ref{prop_game_solving} and Theorem \ref{thrm_external_behavioral_inclusion}.
\end{proof}

\section{Methodology 2: Observers for Concrete Systems}
\label{mthd_concrete_domain_observers}
\begin{figure}	
	\centering
	\includegraphics[width=0.6\textwidth]{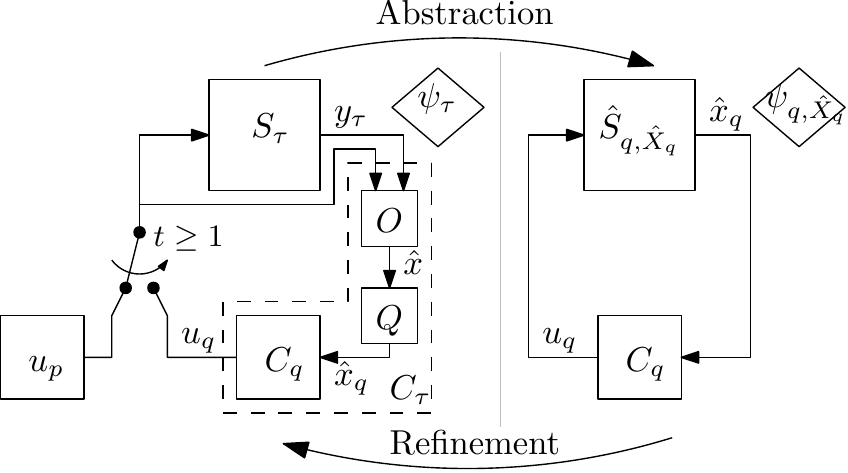}
	\caption{Symbolic control of output-based systems using an observer for the concrete system.}	
	\label{fig_observer_method}
\end{figure}


Observers estimate state values of control systems by observing their input and output sequences.
We consider observers of concrete systems for output-based symbolic control.
We give an informal overview of the methodology and then present it, in details, in the following subsections.
Figure \ref{fig_observer_method} depicts the abstraction and refinement phases of the proposed methodology.
Consider the concrete system $S_\tau$, as introduced in \eqref{eq_concrete_output_system}, and its symbolic model $\sym{S}$, as introduced in \eqref{eq_abstract_output_system}.
We first design an observer $\mathcal{O}$ that estimates the states of $S_{\tau,X_\tau}$ with some upper bound $\epsilon \in \R^+$ for the error between actual states and observed ones.
A state-based symbolic model $\hat{S}_{q, \sym{\hat{X}}}$ is then related to the observed system and used for symbolic controller synthesis.
We show that $\hat{S}_{q, \sym{\hat{X}}}$ can be directly constructed from $\sym{S}$ by inflating each of its states (a state of $\sym{S}$ is a set in $X_\tau$) by $\epsilon$.
The synthesized symbolic controller $\sym{C}$ is finally refined with an interface that uses the observer.

\subsection{Observer Design}
\label{sec_observer_Design}
Let $S_\tau$ be an output-based system and $\sym{S}$ be its symbolic model, as introduced in \eqref{eq_concrete_output_system} and \eqref{eq_abstract_output_system}, respectively, such that $S_\tau \curlyeqprec_Z \sym{S}$, where $Z$ is an OFRR.
Let $Q$ be the underlying FRR of $Z$.
Given a specification $\psi_\tau$, let $(S_\tau, \psi_\tau)$ be a concrete control problem.
We first introduce observers and show how they are composed with $S_\tau$.

\begin{definition}[Observers]
	\label{def_observer}
	Given a precision $\epsilon > 0$, an observer for concrete system $S_\tau$ is a system:
	\begin{equation}
		\nonumber
		\mathcal{O} := (\hat{X}, \hat{X}, \hat{U}, \underset{\mathcal{O}}{\longrightarrow}, \hat{X}, id_{\hat{X}}),
	\end{equation}
	where 
	$\hat{X} := X_\tau$, $\hat{U} := U_\tau \times Y_\tau$, and $\underset{\mathcal{O}}{\longrightarrow}$ is defined such that the following holds for all $x_{0} \in X_\tau$ and all $\hat{x}_0 \in \hat{X}$:
	\begin{align}
		\label{eq_meth1_bounded_observ}
		&\forall r_{int} \in \mathsf{RUNS}_{int}(S_\tau^{(\{x_{0}\})}) \; \forall \hat{r}_{int} \in \mathsf{RUNS}_{int}(\mathcal{O}^{(\{\hat{x}_0\})}) \; \\ \nonumber
		&\;\;(\pi_{U_\tau}(r_{int}) = \pi_{U_\tau}(\hat{r}_{int}) \implies \\ \nonumber
		&\;\;\;\;\;\;\;\forall n \geq 1 \; (\Vert \Last{\hat{r}_{int}(n)} - \Last{r_{int}(n)} \Vert \leq \epsilon)).
	\end{align} 
\end{definition}

Note that, for any linear time-invariant control systems, it is always possible to construct $\mathcal{O}$ by embedding a Luenberger observer with a suitable gain inside it \cite{Franklin:1997:DCD:550726}.
Additionally, for some classes of nonlinear systems, one can utilize high-gain observers \cite{RNC:RNC3051}.
We define the \emph{observed system} $\hat{S}_\tau$ as the system resulting from composing the observer $\mathcal{O}$ to the sampled-data system $S_\tau$ as follows:
\begin{equation}
	\nonumber
	\hat{S}_\tau := \mathcal{O} \triangleleft S_\tau,
\end{equation}
where $\triangleleft$ denotes the observation composition introduced in Definition \ref{def_observ_composition} and the output set of $\hat{S}_\tau$ is consequently equals to $\hat{X} = X_\tau$.
Here, $\hat{S}_\tau$ coincides with its state-based system $\hat{S}_{\tau, X_\tau}$ version and we use them interchangeably.

In Definition \ref{def_observer}, the distance between the runs of $S_\tau$ and those of $\hat{S}_\tau$ is always upper bounded by $\epsilon$ after the first sampling period.
We synthesize symbolic controllers to solve $(S_\tau, \psi_\tau)$ only after the first sampling period.
In Subsection \ref{ssec_observers_blind_period}, we show how to handle the first sampling period.

\subsection{A symbolic model for $\hat{S}_\tau$}
We approximate $\mathcal{O}$ with a static perturbation map, denoted by $\widetilde{\mathcal{O}}: X_{\tau} \rightrightarrows X_{\tau}$ ($\rightrightarrows$ denotes set-valued mapping), such that its perturbation is upper bounded by $\epsilon$.
Formally, we define map $\widetilde{\mathcal{O}}$ as follows for any $x \in X_\tau$:
\begin{equation}
	\nonumber
	\widetilde{\mathcal{O}}(x) := \{ \tilde{x} \in X_\tau \;\vert\; \Vert x - \tilde{x} \Vert \le \epsilon \}.
\end{equation}
One can simply show that $B(\mathcal{O} \triangleleft S_\tau) \subseteq B(\widetilde{\mathcal{O}} \circ S_{\tau, X_\tau})$.
Now let us recall the symbolic model $\sym{S}$ of $S_\tau$.
Note that the elements of $\sym{X}$ are disjoint subsets of $X_\tau$.
A symbolic model for $\hat{S}_\tau$ is constructed by inflating each state $\sym{x} \in \sym{\hat{X}}$ of $\sym{S}$ by $\epsilon$.
Formally, we denote by $\sym{\hat{S}}$ the symbolic model of the observed system $\hat{S}_\tau$ and we define it as follows:
\begin{equation}
	\label{eq_s_q_inflated}
	\sym{\hat{S}}:=(\sym{\hat{X}}, \sym{\hat{X}}, \sym{U}, \underset{\hat{q}}{\longrightarrow}, \sym{\hat{X}}, id_{\sym{\hat{X}}}),
\end{equation}
where $\sym{\hat{X}} := \{\underset{x_\tau \in \sym{x}}{\bigcup}\widetilde{\mathcal{O}}(x_\tau) \;\vert\; \sym{x} \in \sym{X} \}$,
and $(\sym{\hat{x}}, \sym{u}, \sym{\hat{x}}') \in \underset{\hat{q}}{\longrightarrow}$ if there exist 
$x \in \sym{\hat{x}}$ and $x' \in \sym{\hat{x}}'$ such that 
$((x,\hat{x}),\sym{u},(x',\hat{x}')) \in \underset{\hat{\tau}}{\longrightarrow}$ 
for some $\hat{x},\hat{x}' \in \hat{X}$,
and $\underset{\hat{\tau}}{\longrightarrow}$ is the transition relation of $\hat{S}_\tau$.
Notice that $\sym{\hat{S}}$ also coincides with its state-based version $\hat{S}_{q,\sym{\hat{X}}}$ and we use them interchangeably.

\begin{remark}
	\label{rmk_observed_states_one2one}
	For any $\epsilon > 0$, the elements of $\sym{\hat{X}}$ form a cover of $X_\tau$ and its elements have one-to-one correspondence with the partition elements of $\sym{X}$.
\end{remark}

Now we derive a version of the given specification $\psi_\tau$ to be used later for controller synthesis.
First, a state-based abstract specification $\psi_{q,\sym{X}}$ is derived using maps $Z$ and $\sym{H}$ as follows: $\psi_{q,\sym{X}} := \sym{H}^{-1}(Z(\psi_\tau))$.
Here, we abuse the notation and apply $Z$ and $\sym{H}^{-1}$ to elements of state sequences in $\psi_\tau$.
Then, we define a map $\widetilde{\mathcal{O}}_q: \sym{X} \to \sym{\hat{X}}$ that accepts a partition element $\sym{x} \in \sym{X}$ and translates it to its corresponding cover element $\sym{\hat{x}} \in \sym{\hat{X}}$.
Using $\widetilde{\mathcal{O}}_q$, any state-based abstract specification $\psi_{q,\sym{X}}$ can be translated to an abstract specification $\psi_{q, \sym{\hat{X}}}$ as follows: $\psi_{q, \sym{\hat{X}}} := \widetilde{\mathcal{O}}_q(\psi_{q,\sym{X}})$.
Finally, we have $(\hat{S}_{q,\sym{\hat{X}}}, \psi_{q,\sym{\hat{X}}})$ as an observed-based abstract control problem and its construction is depicted with steps $\mathtt{(1)}$ to $\mathtt{(4)}$ in Fig. \ref{fig_observer_method_detailed}.

\begin{remark}
	\label{rmk_hints_for_choosing_epsilon}
	Although observer $\mathcal{O}$ is designed for $S_\tau$, the choice of $\epsilon$ should be based on states set $X_q$ in $S_q$.
	Selecting a larger value of $\epsilon$ increases the nondeterminism of transitions of $\hat{S}_{q,\hat{X}_q}$ making control problem $(\hat{S}_{q,\hat{X}_q}, \psi_{q,\hat{X}_q})$ unsolvable.
\end{remark}

\subsection{Controller Synthesis and Refinement}
\label{ssec_observers_synth_refine}
\begin{figure}	
	\centering
	\includegraphics[width=0.6\textwidth]{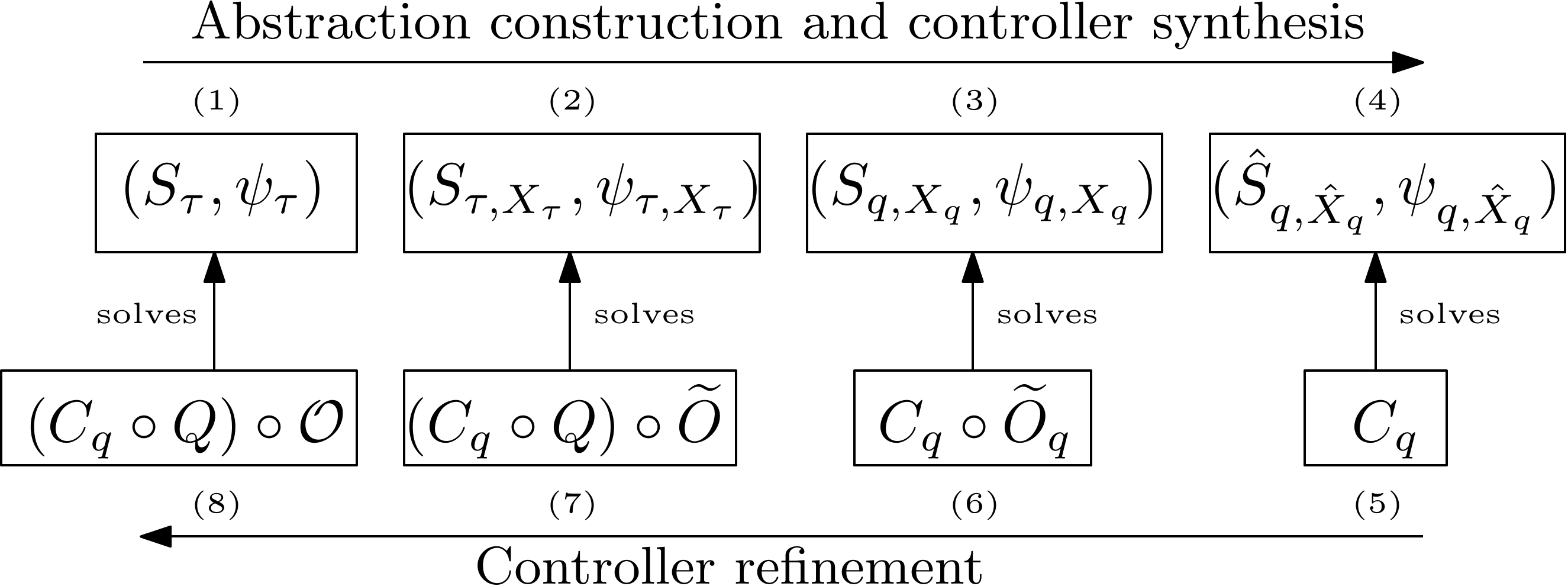}
	\caption{
		Construction of different control problems and their controllers for the observer-based methodology.
	}	
	\label{fig_observer_method_detailed}
\end{figure}


In the previous subsection, we demonstrated how $(S_\tau,\psi_\tau)$ is translated to $(\hat{S}_{q,\hat{X}_q}, \psi_{q,\hat{X}_q})$, as depicted in Fig. \ref{fig_observer_method_detailed}.
We know from Corollary \ref{corr_internal_behavior_implies_extrnal_behavior} that a controller designed for control problem in step $\mathtt{(2)}$ can be refined to solve control problem in step $\mathtt{(1)}$.
Here, we rely on two facts: the behavior of $S_{\tau,X_\tau}$ is the internal behavior of $S_\tau$, and $B(\mathcal{O} \triangleleft S_\tau) \subseteq B(\widetilde{\mathcal{O}} \circ S_{\tau, X_\tau})$.
The control problem in step $\mathtt{(3)}$ is a symbolic representation of the control problem in step $\mathtt{(2)}$ using the FRR $Q$.
The results from \cite{G.Reissig_etal_FRR_TAC} apply directly and any controller designed to solve the control problem in step $\mathtt{(3)}$ can be refined to solve the control problem in step $\mathtt{(2)}$ using $Q$ as a static quantization map.
The only missing link is how a controller designed to solve the control problem in step $\mathtt{(4)}$ is refined to solve the control problem in step $\mathtt{(3)}$.
Note that in \eqref{eq_s_q_inflated}, we designed $\hat{S}_{q,\hat{X}_q}$ by inflating the states of $S_{q,X_q}$ using the perturbation map $\widetilde{\mathcal{O}}$.
Hence, we can use the results in \cite[Theorem VI.4]{G.Reissig_etal_FRR_TAC} to ensure the behavioral inclusion when refining the controller designed for the control problem in step $\mathtt{(4)}$.
We introduce a version of \cite[Theorem VI.4]{G.Reissig_etal_FRR_TAC} adapted to our notation.

\begin{theorem}
	\label{thrm_meth1_controllers_synth_and_refine}
	Let $(S_{q, X_q}, \psi_{q,X_q})$ be an abstract control problem .
	Consider an abstract observer-based control problem $(\hat{S}_{q, \hat{X}_q}, \psi_{q,\hat{X}_q})$ constructed using the map $\widetilde{\mathcal{O}}_q$.
	If a controller $\sym{C}$ solves $(\hat{S}_{q, \hat{X}_q}, \psi_{q,\hat{X}_q})$, then the controller $\sym{C} \circ \widetilde{\mathcal{O}}_q$ solves $(S_{q, X_q}, \psi_{q,X_q})$.
\end{theorem}
\begin{proof}
	The proof is very similar to that of \cite[Theorem VI.4]{G.Reissig_etal_FRR_TAC} and is omitted here due to lack of space.
\end{proof}


\subsection{The First Sampling Period}
\label{ssec_observers_blind_period}
The symbolic controller $\sym{C}$ is only valid after the first sampling period.
One solution to ensure that the system is ready for $\sym{C}$ for times $t \ge \tau$, is to choose an input $u_p \in U_\tau$ and an initial state set $X_p \in X_\tau$ satisfying
\begin{align}
	\label{eq_obersever_blindperiod_cond}
	\forall x_0 \in X_p \;  & \forall x_\tau \in \Post{x_0}{u_p}{S_\tau} \; \exists x_q \in \mathcal{D}(\sym{C}) \text{ s.t. } \\ \nonumber
	                        &x_\tau \in x_q,
\end{align}
where $\mathcal{D}$ extracts the controller's domain as introduced in Definition \ref{def_controller_domain}.
Condition \eqref{eq_obersever_blindperiod_cond} ensures that states at times $t \ge \tau$ remains in $\mathcal{D}(\sym{C})$.
We then need to solve a special control problem $(S_{\tau}^{(X_p)}, \psi_p)$, where $\psi_p$ is defined as follows:
\begin{equation}
	\nonumber
	\psi_p := 
	\begin{cases} 
		\nonumber
		\Safe_{[0,1]}(H_{\tau}( X_D )), 	& \mbox{if } X_p \subseteq  \mathcal{D}(\sym{C})\\ 
		\nonumber
		\Reach_{[0,1]}(H_{\tau}(X_D )), 	& \mbox{if } \mathcal{D}(\sym{C}) \subset  X_p\\ 
	\end{cases},
\end{equation}
and $X_D := \underset{\sym{x} \in \mathcal{D}(\sym{C})}{\bigcup}  \sym{x}$.
The selection of $X_p$ is critical and depends on the dynamics of $\Sigma$.
A good strategy is to start with $X_p = X_D$ and expand (or shrink) it until condition \eqref{eq_obersever_blindperiod_cond} is met for some input $u_p$.
We discuss this again with an example in Section \ref{sec_examples}.

\section{Methodology 3: Constructing Detectors for Symbolic Models}
\label{meth_detector}
\begin{figure}	
	\centering
	\includegraphics[width=0.6\textwidth]{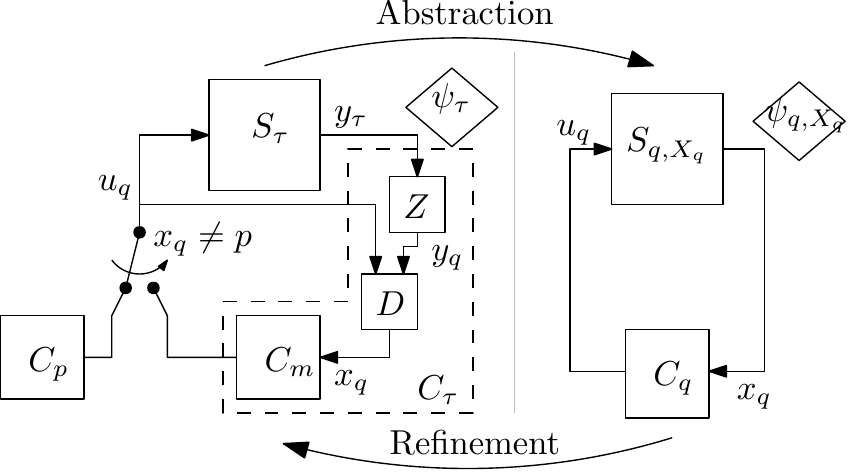}
	\caption{
		Output-feedback symbolic control using detectors.
	}	
	\label{fig_detector_method}
\end{figure}

We revise the notion of detectability of non-deterministic finite transition systems (NFTS) \cite{Zhang2020}, and use it to design detectors for $\sym{S}$.
First, we introduce non-deterministic finite automata (NFA).
Here, a system is intuitively called detectable if one can use sufficiently long input sequences and their corresponding output sequences to determine the current and all subsequent states of the system. 
We first show how to construct detectors to identify, in finite-time, the current and all subsequent states of detectable symbolic models.
We then synthesize symbolic controllers and refine them to enforce the given specifications on original systems.
The method is depicted schematically in Fig. \ref{fig_detector_method} and we summarize it as follows: 
1) construct an abstract control problem $(\sym{S}, \sym{\psi})$ from the concrete control problem $(S_\tau, \psi_\tau)$, where $\sym{S} \preccurlyeq_{Z} S_\tau$;
2) verify the detectability of $\sym{S}$;
3) if $\sym{S}$ is detectable, then design a detector $D$ to detect its state at the current time step;
4) use the state-based system $S_{q,\sym{X}}$ to synthesize a symbolic controller $C_q$ that is wrapped with some routing signals in a symbolic controller $C_m$; 
5) refine symbolic controller $C_m$ using $Z$ and $D$ as interfaces to controller $C_{\tau}$; and finally
6) since $D$ requires a priori known finite time to start detecting the states of $\sym{S}$, an open-loop controller $C_p$ is designed to keep the system in the domain of $C_q$, and a signal $p$ is required to switch between $C_p$ and $C_{\tau}$.


\old{\subsection{Additional Notation}}
\old{We introduce extra notations needed only for this section.}
\old{For better readability of the paper, we preferred to introduce these extra notations here rather than in the beginning of the paper.}
\old{Consider a system $S := (X, X_0, U, \longrightarrow, Y, H)$.
For all $x\in X$ and $\a\in U^*$ such that $|\a|\ge 1$, $x'\in X$ is called an
$\a$-successor of $x$, if there exist states $x_0,\dots,x_{|\a|}\in X$ such that $x_0=x$, $x_{|\a|}=x'$,
and $(x_i,\a_i,x_{i+1})\in\longrightarrow$ for all integers $0\le i\le |\a|-1$.
The set of $\a$-successors of a state $x\in X$ (resp., a subset $X'\subset X$) is denoted by $\Post{x}{\a}{}$
(resp., $\Post{X'}{\a}{}:=\cup_{x\in X'}\Post{x}{\a}{}$).
For all $x\in X$, $\a\in U^*$ and $\b\in Y^*$ such that $|\a|+1=|\b|$, $x'\in X$ is called an
$(\a,\b)$-successor of $x$, if there exist states $x_0,\dots,x_{|\a|}\in X$ such that $x_0=x$, $x_{|\a|}=x'$,
$H(x_{|\a|})=\b_{|\a|}$, and $H(x_i)=\b_i$ and $(x_i,\a_i,x_{i+1})\in\longrightarrow$ for all integers $0\le i\le |\a|-1$.
The set of $(\a,\b)$-successors of a state $x\in X$ (resp., a subset $X'\subset X$) is denoted by $\Post{x}{\a}{\b}$
(resp., $\Post{X'}{\a}{\b}:=\cup_{x\in X'}\Post{x}{\a}{\b}$).}

\movedagain{\begin{definition}
	An NFA $\A$ is a tuple $\A := (\mathcal{Q},\mathbf{\Delta},\d,q_0,F)$,
where $\mathcal{Q}$ is a finite set of states, $\mathbf{\Delta}$ is a finite set of labels (which is an alphabet), $\d\subset \mathcal{Q}\times \mathbf{\Delta}\times \mathcal{Q}$ is the transition relation, $q_0\in \mathcal{Q}$ is the initial state, and $F\subset \mathcal{Q}$ is a set of final states. 
\end{definition}}

The transition relation $\d$ of NFA $\A$ is extended to $\d^*\subset \mathcal{Q}\times \mathbf{\Delta}^*\times \mathcal{Q}$ in the usual way:
for all $q,q'\in \mathcal{Q}$, $(q,\e,q')\in\d^*$ iff $q=q'$; and for all $q,q'\in \mathcal{Q}$ and $\sigma_0\dots\sigma_{n-1}\in\mathbf{\Delta}^*
\setminus\{\e\}$, $(q,\sigma_0\dots\sigma_{n-1},q')\in\d^*$ iff there exists $q_1,\dots,q_{n-1}\in \mathcal{Q}$
such that $(q,\sigma_0,q_1), (q_1,\sigma_1,q_2),\dots,(q_{n-1},\sigma_{n-1},q')\in\d$.
Hereinafter, we use $\d$ to denote $\d^*$, as no confusion shall occur.
A state $q\in \mathcal{Q}$ is said to be reachable from a state $q'\in \mathcal{Q}$, if there exists $\sigma\in\mathbf{\Delta}^*$ such that
$(q',\sigma,q)\in\d$. A state $x\in \mathcal{Q}$ is called reachable from a subset $\mathcal{Q}'$ of $\mathcal{Q}$, if $x$ is reachable from
some states of $\mathcal{Q}'$.
A sequence $q_0,\dots,q_n\in \mathcal{Q}$ is called a path, if there exist $\sigma_0,\dots,\sigma_{n-1}\in\mathbf{\Delta}$
such that $(q_0,\sigma_0,q_1),\dots,(q_{n-1},\sigma_{n-1},q_n)\in\d$.
A path $q_0,\dots,q_n\in \mathcal{Q}$ is called a cycle, if $q_0=q_n$.
\oldagain{Note that $\d\subset \mathcal{Q}\times\mathbf{\Delta}^*\times \mathcal{Q}$ is equivalently represented as a map $\d:\mathcal{Q}\times\mathbf{\Delta}^*\longrightarrow 2^{\mathcal{Q}}$: for all $q,q'\in \mathcal{Q}$ and $\sigma\in \mathbf{\Delta}^*$, $(q,\sigma,q')\in\d$ iff $q'\in\d(q,\sigma)$.}

We borrow the concept of limit points from the theory of cellular automata \cite{KariCALectureNote} and use it for NFAs.
Limit points are defined as the points that can be visited at each time step.
If one regards an NFA $\A$ as a system in which each state is initial, and regard each state of $\A$ as a point, then limit points are exactly the states reachable from some cycles. 
The limit set of $\A$ consists of limit points and we denote it by $LP(\A)$. 

\subsection{Detectability of Symbolic Models}
Consider a concrete control problem $(S_{\tau}, \psi_{\tau})$ and its abstract control problem $(\sym{S}, \sym{\psi})$ such that $\sym{S} \preccurlyeq_{Z} S_{\tau}$, for some OFRR $Z$, and $\sym{\psi}$ is constructed as introduced in \eqref{eq_from_concrete_specs_to_abstract_specs}.
We first introduce the concept of detectability\old{ (a.k.a. arbitrary-experiment detectability \cite{Zhang2020})} for symbolic models.

\begin{definition}[Detectability of Symbolic Models]
	\label{def_detectability_NFTS}
	A symbolic model $S_q$, as defined in \eqref{eq_abstract_output_system}, is said to be detectable if there exists $N \in \R^+$ such that for all input sequences $\a\in U^{*}$, $\vert \a \vert \ge N$, and all output sequences $\b\in Y^{*}$, $\vert \b \vert = |\a|+1$, we have that $\vert \Post{X_q}{\a}{\b} \vert \le 1$.
\end{definition}

Verifying the detectability of $\sym{S}$ is essential in the current methodology.
We introduce Algorithm \ref{alg1:detectability_NFTS} that takes $\sym{S}$ as input, and returns NFA $\A$ which is used to \oldagain{infer}\newagain{check} the detectability of $\sym{S}$\oldagain{ and to design a detector $D$ for its states}.

\begin{algorithm}
	\label{alg1:detectability_NFTS}
	Receive a symbolic model $\sym{S} := (X_q,X_q,U_q,\underset{q}{\longrightarrow},Y_q,H)$, and initiate an NFA $\A := (\mathcal{Q},\mathbf{\Delta},\d,q_0,F)$, where
	$\mathcal{Q} :=\{\diamond\}$, $\diamond$ is a dummy symbol, $\mathbf{\Delta} := \d := F := \emptyset$, and $q_0 := \diamond$.
	$\mathcal{Q}_1 := \emptyset$, $\mathcal{Q}_2 := \emptyset$. 
	Let $\phi$ be a dummy symbol not in $Y_q$.
	\begin{enumerate}
		\item \label{item1:detectability_NFTS}
			For each $y\in Y_q$, denote $X_y:=\{x\in X_q|H_q(x)=y\}$,
			\begin{enumerate}
				\item if $|X_y|=1$, then
					$\mathcal{Q}_1:=\mathcal{Q}_1\cup\{X_y\}$, $\mathbf{\Delta}:=\mathbf{\Delta}\cup\{(\phi,y)\}$, $\d:=\d\cup\{(\diamond,(\phi,y),
					X_y)\}$,
				\item else if $|X_y|>1$, then $\mathcal{Q}_1:=\mathcal{Q}_1\cup\{Z\subset X_y||Z|=2\}$,
					$\mathbf{\Delta}:=\mathbf{\Delta}\cup\{(\phi,y)\}$, for each $Z\subset X_y$ satisfying that $|Z|=2$,
					$\d:=\d\cup\{(\diamond,(\phi,y),Z)\}$.
			\end{enumerate}
			${\mathcal{Q}}:={\mathcal{Q}}\cup \mathcal{Q}_1$, $\mathcal{Q}_2:=Q_2\cup \mathcal{Q}_1$, $\mathcal{Q}_1:=\emptyset$.
		\item \label{item2:detectability_NFTS}
			If $\mathcal{Q}_2=\emptyset$, stop. 
			Else, for each $q_2\in \mathcal{Q}_2$, denote $y_0:=H_q(x)$, where $x\in q_2$,
			for each $u\in U_q$ and each $y\in Y_q$,
			\begin{enumerate}
				\item if $|\post_u^{y_0y}(q_2)|=1$, then
					$\mathbf{\Delta}:=\mathbf{\Delta}\cup\{(u,y)\}$, $\d:=\d\cup\{(q_2,(u,y),\post_u^{y_0y}(q_2))\}$,
					if $\post_u^{y_0y}(q_2)\notin {\mathcal{Q}}$ then $\mathcal{Q}_1:=\mathcal{Q}_1\cup\{\post_u^{y_0y}(q_2)\}$,
				\item else if $|\post_u^{y_0y}(q_2)|>1$, then $\mathbf{\Delta}:=\mathbf{\Delta}\cup\{(u,y)\}$, 
					for each $Z\subset\post_u^{y_0y}(q_2)$ satisfying $|Z|=2$, 
					$\d:=\d\cup\{(q_2,(u,y),Z\}$,
					if $Z\notin {\mathcal{Q}}$ then $\mathcal{Q}_1:=\mathcal{Q}_1\cup\{Z\}$.
			\end{enumerate}
			${\mathcal{Q}}:={\mathcal{Q}}\cup \mathcal{Q}_1$, $\mathcal{Q}_2:=\emptyset$, $\mathcal{Q}_2:=\mathcal{Q}_1$, $\mathcal{Q}_1:=\emptyset$.
		\item \label{item3:detectability_NFTS}
			Go to Step \eqref{item2:detectability_NFTS}. (Since $X_q$,$U_q$, and $Y_q$ are finite, the algorithm will terminate.)
	\end{enumerate}
\end{algorithm}


Let $\A$ be the NFA resulting from Algorithm \ref{alg1:detectability_NFTS} after setting $\sym{S}$ as input. 
The smallest natural number $T_t$ such that each pair of input sequence of length $T_t$ and output sequence of length $T_t+1$ changes the initial state of the $\A$ to a state of its limit set is called the \emph{transient period}.
More precisely, we denote by $T_t$ the transient period of $\sym{S}$ and define it as follows:

\begin{equation}
	\begin{split}
		T_t :=& \min\left\{ t\in\N \;\vert\;
		\forall u_1,\dots,u_t\in U_q \;
		\forall y_0,\dots,y_t\in Y_q \; \right.\\
		&\qquad (\d(\diamond,(\phi,y_0)(u_1,y_1)\dots(u_t,y_t)) \ne \emptyset \implies \\
		&\qquad \d(\diamond,(\phi,y_0)(u_1,y_1)\dots(u_t,y_t)) \subseteq LP(\A) \left.)\right\}.
	\end{split}
	\nonumber
\end{equation}

The next theorem provides a tool to check the detectability of $\sym{S}$ and provides a time index $T_t$ after which one can identify the states of the system.

\begin{theorem}\label{thm1:detectability_NFTS}
	Let $\sym{S}$ be a symbolic model as introduced in \eqref{eq_abstract_output_system}.
	Let $\A$ be the NFA resulting from running Algorithm \ref{alg1:detectability_NFTS} with $\sym{S}$ as input and $T_t$ be its transient period.
	Then, 
	\begin{itemize}
		\item [(i)] $\sym{S}$ is detectable iff in $\A$, each state reachable from some cycle is a singleton, and
		\item [(ii)] if $\sym{S}$ is detectable, then for all input sequences $\a\in U_q^{*}$, $\vert \a \vert \ge T_t$, and all output sequences $\b\in Y_q^{*}$, $\vert \b \vert = |\a|+1$, we have that $\vert \Post{X_q}{\a}{\b} \vert \le 1$.
	\end{itemize}
\end{theorem}
\begin{proof}
	\new{The proof of (i) is given in \cite[Theorem 8.1]{Zhang2020}.}
	\old{First, we show (i).
    Note that the number of states of the NFA $\A=(\mathcal{Q},\mathbf{\Delta},\d,\diamond,\emptyset)$ that 
    Algorithm \ref{alg1:detectability_NFTS} returns is no greater than $N:=|X|(|X|-1)/2+|X|+1$.
	
	Assume that the system is not detectable.
	Then there exists an input sequence $\a\in U^*$ of length greater than $N$ and an output sequence $\b\in Y^*$ of length $|\a|+1$ such that $\post_\a^\b(X_0)$ has cardinality bigger than $1$, and for each integer $0\le i <|\a|$, $\post_{\a[0,i]}^{\b[0,i+1]}(X_0)$ is not empty.
    Choose $q_{|\a|}\subset \post_{\a}^{\b}(X_0)$ satisfying $|q_{|\a|}|=2$.
    For all $j=2,\dots,|\a|$, choose nonempty $q_{|\a|-j+1}\subset \post_{\a[0,|\a|-j]}^{\b[0,|\a|-j+1]}(X_0)$ satisfying $|q_{|\a|-j+1}|=2$ if $|\post_{\a[0,|\a|-j]}^{\b[0,|\a|-j+1]}(X_0)|\ge 2$, and for each $x\in q_{|\a|-j+2}$, $(x',\a(|\a|-j+1),x)\in\longrightarrow$ for some $x'\in q_{|\a|-j+1}$.
    Choose $q_{0}\subset X_0$ satisfying $|q_{0}|=2$ if $|X_0|\ge 2$, and for each $x\in q_1$, $(x',\a(0),x)\in\longrightarrow$ for some $x'\in q_{0}$.
    Then $q_{0},\dots,q_{|\a|}$ are states of the NFA $\A$, $q_{i}$ is reachable from $q_{i-1}$ for all integers $0< i\le|\a|$, and there exist $0\le j< k\le|\a|$ such that $q_j=q_k$ by the pigeon-hole principle \cite{Pigeonhole}.
    Hence, $q_{|\a|}$ is reachable from the cycle $q_j,\dots,q_k$, i.e., $q_{|\a|}$ is a limit point of $\A$ with cardinality $>1$. Hence the ``if'' part holds.

	Assume that the system is detectable. 
	We are given an arbitrary state $q$ of the NFA $\A$ and assume that $q$ is reachable from a cycle.
	Then there exist states $q_1,\dots,q_p\in \mathcal{Q}$ such that $q_1$ is reachable from $\diamond$, $q_{i+1}$ is reachable from $q_i$ for all integers $1\le i<p$, $q$ is reachable from $q_p$, and $q_j=q_k$ for some integers $1\le j<k\le p$.
	By the definition of $\A$, there exist $u_1,\dots,u_p\in U$ and $y_0,\dots,y_p\in Y$ such that $(\diamond,(\phi,y_0),q_1), (q_1,(u_1,y_1),q_2),\dots,(q_p,(u_p,y_p),q)\in\d$. 
	Since the system is detectable, for sufficiently large integer $\tilde{n}$, the set $\post_{ \mathbb{U}_{\tilde{n}}}^{\mathbb{Y}_{\tilde{n}}}(X_0)$ has cardinality $\le 1$, where $\mathbb{U}_{\tilde{n}}=u_1\dots u_{j-1}(u_{j}\dots u_{k-1})^{\tilde{n}}u_k\dots u_p$, $\mathbb{Y}_{\tilde{n}}=y_0y_1\dots y_{j-1}\\(y_{j}\dots y_{k-1})^{\tilde{n}}y_k\dots y_p$, and $(\cdot)^{\tilde{n}}$ means the concatenation of $\tilde{n}$ copies of $\cdot$. 
	We also have $\emptyset\ne q\subset \post_{ \mathbb{U}_{\tilde{n}}}^{\mathbb{Y}_{\tilde{n}}}(X_0)$, then $|q|=1$, which shows that (i) holds.
	}

	\new{The proof of (ii) is given in \cite[Proposition 8.1]{Zhang2020}.}
	\old{Now, we show (ii).
    Since the symbolic system is detectable, all limit points of $\A$ are singletons by part (i).
    To prove part (ii), we only need to show that for all input sequences $\a\in U^*$ of length $T_t$ and output sequences $\b\in Y^*$ of length $|\a|+1$, $\d\left( \diamond, (\phi,\b(0))(\a(0),\b(1))\dots (\a(T_t-1),\b(T_t)) \right)$ has cardinality less than or equal to $1$.
    Note that the union of elements of $\d( \diamond, (\phi,\b(0))(\a(0),\b(1))\dots (\a(T_t-1),\b(T_t)) )$ equals $\post_{\a}^{\b}(X_0)$.
	Suppose on the contrary that there exist $u_1,\dots,u_{T_t}\in U$, $y_0,\dots,y_{T_t}\in Y$, $q_0,\dots,q_{T_t},q_0',\dots,q_{T_t}'\subset X$ such that $0<|q_i|\le 2$ and $0<|q_i'|\le 2$ for all integers $0\le i\le T_t$, $|q_{T_t}|=|q_{T_t}'|=1$, $q_{T_t}\ne q_{T_t}'$, $(\diamond,(\phi,y_0),q_0), (\diamond,(\phi,y_0),q_0')\in\d$, $(q_j,(u_{j+1},y_{j+1}),q_{j+1}),(q_j',(u_{j+1},y_{j+1}),q_{j+1}')\in\d$ for all integers $0\le j\le T_t-1$. 
	For each integer $0\le i\le T_t$, choose $x_i\in q_i$, $x_i'\in q_i'$ such that $(x_0,u_1,x_1)\in\longrightarrow$, $(x_0',u_1,x_1')\in\longrightarrow$, $\dots$, $(x_{T_t-1},u_{T_t},x_{T_t})\in\longrightarrow$, $(x_{T_t-1}',u_{T_t},x_{T_t}')\in\longrightarrow$.
	Denote $q_i'':=\left\{ x_i,x_i' \right\}$, for all $i=0,\dots,T_t$. 
	Then, $|q_{T_t}''|=2$.
    If $|q_0''|=1$, and there exists $x\in X\setminus q_0''$ satisfying $H(x)=H(x_0)$, add at most one such $x$ into $q_0''$, i.e., $q_0'':=q_0''\cup\{x\}$, and in this case $|q_0''|=2$.
    Then, we modify $q_1'',\dots,q_{T_t-1}''$ successively as follows.
    For all integers $j=1,\dots, T_t-1$, if $|q_j''|=1$, and there exists $x\in X\setminus q_j''$ satisfying $H(x)=H(x_j)$, and $(x',u_j,x)\in\longrightarrow$ for some $x'\in q_{j-1}''$, we add at most one such $x$ into $q_j''$, i.e., $q_j'':=q_j''\cup\{x\}$, and in this case $|q_j''|=2$.
    As a result, $q_0'',\dots,q_{T_t}''$ are states of $\A$, and $(\diamond,(\phi,y_0),q_0'')\in\d$, $(q_j'',(u_{j+1},y_{j+1}),q_{j+1}'')\in\d$ for all integers $0\le j\le T_t-1$.
    By the definition of transient period, $q_{T_t}''$ is a limit point of $\A$. But $q_{T_t}''$ has cardinality $2$, which contradicts that all limit points of $\A$ are singletons.
    Hence, (ii) holds. 	
	}
\end{proof}


\subsection{Controller Synthesis and Refinement}
Consider a \newagain{detectable} symbolic model $\sym{S}$\oldagain{ and assume it is detectable}.
We show how to \oldagain{use $\A$ to }design a detector for \oldagain{$\sym{S}$ starting from time index $T_t$}\newagain{it}.
Let $\A := (\mathcal{Q}, \mathbf{\Delta}, \delta, \{\diamond\}, F)$ be the NFA resulting from Algorithm \ref{alg1:detectability_NFTS} with $\sym{S}$ as input. 
We introduce the detector system as follows:
\begin{equation}
	\label{eq_detector_system}
	D := (X_D, X_{D,0}, \sym{U} \times \sym{Y}, \underset{D}{\longrightarrow}, Y_D, H_D),
\end{equation}
where
\begin{itemize}
	\item $X_D := \sym{X} \times \mathcal{Q} \times \{0,1\}$;
	\item $X_{D,0} := \{ (x_q, \diamond, 0) \;\vert\; x_q \in X_q \}$;
	\item $\underset{D}{\longrightarrow} := \{((x_q, q, 0),(u_q,y_q),(x_q', q', 1)) \;\vert\; (x_q,u_q,x_q') \in \underset{q}{\longrightarrow} \;\land\; (q,(u_q,y_q),q') \in \delta \;\land\; \vert q' \vert \leq 1 \} \cup \{((x_q, q, f),(u_q,y_q),(x_q', q', f)) \;\vert\; (x_q,u_q,x_q') \in \underset{q}{\longrightarrow} \;\land\; (q,(u_q,y_q),q') \in \delta \;\land\; (\vert q' \vert > 1 \;\lor\; f = 1 )\}$;
	\item $Y_D := \sym{X} \cup \{p\}$, where $p$ is a dummy symbol denoting incomplete detection of the state of $\sym{S}$; and
	\item $H_D$ is defined as follows:
	\begin{equation*}
		H_D((x_q, q, f)) := 
			\begin{cases}
				x_q 	& f = 1\\
				p 		& f = 0
			\end{cases}
			.
	\end{equation*}
\end{itemize}

\begin{remark}
	\label{rmk_detector_output_after_Tt}
	After $T_t$ sampling periods of providing inputs and observations of $\sym{S}$ to $D$, we have that:
	\begin{itemize}
		\item[(1)] $H_D(x_D) \neq p$, for any $x_D \in X_D$,
		\item[(2)] $H_D(x_D)$ provides the detected current state of $\sym{S}$, and
		\item[(3)] $B_{ext}(D \triangleleft (Z \circ S_\tau)) = B_{int}(\sym{S})$.
	\end{itemize}	
\end{remark}


A controller \old{$C_q := (X_{\sym{C}}, X_{\sym{C},0}, U_{\sym{C}}, \underset{C_q}{\longrightarrow}, Y_{\sym{C}}, H_{\sym{C}})$}\new{$C_q$, as defined in Definition \ref{def_controller},} can be synthesized to solve $(\sym{S}, \sym{\psi})$, as discussed in \ref{ssec_controller_synthesis}.
\new{Then, using Theorems \ref{thrm_external_behavioral_inclusion} and Remark \ref{rmk_detector_output_after_Tt}(3), $C_q$ is refined using the detector system $D$ and the static map $Z$ as interface, as shown in Fig. \ref{fig_detector_method}.}
We only need to encapsulate $C_q$, in the following system $C_m$, to handle the detection signal $p$:

\begin{equation}
	C_m := (X_{C_m}, X_{C_m,0}, U_{C_m}, \underset{C_m}{\longrightarrow}, Y_{C_q} \cup \{ \kappa \}, H_{C_m}),
\end{equation}
where
\begin{itemize}	
	\item $\kappa$ is a dummy symbol for unavailability of control inputs;
	\item $X_{C_m} := X_{C_q} \cup \{0,1\}$;
	\item $X_{C_m,0} := \{ (x_{C_q},0) \;\vert\; x_{C_q} \in X_{C_q,0} \}$;
	\item $U_{C_m} := U_{C_q} \cup \{p\}$, where $p$ is the symbol from \eqref{eq_detector_system};
	\item $\underset{C_m}{\longrightarrow} := \{((x_{C_q}, 0),u_{C_m},((x'_{C_q}, 1)) \;\vert\; (x'_{C_q},u_{C_m},x'_{C_q}) \in \underset{C_q}{\longrightarrow} \;\land\; u_{C_m} \neq p \} \cup \{((x_{C_q}, 1),u_{C_m},((x'_{C_q}, 1)) \;\vert\; (x'_{C_q},u_{C_m},x'_{C_q}) \in \underset{C_q}{\longrightarrow} \;\land\; u_{C_m} \neq p \} \cup \{((x_{C_q}, 0),u_{C_m},((x_{C_q}, 0))  \;\vert\; u_{C_m} = p \}$; and
	\item $H_{C_m}((x_{C_q}, f)) :=
			\begin{cases}
				H_{C_q}(x_{C_q})	& f = 1\\
				\kappa 				& f = 0
			\end{cases}.$
\end{itemize}


\oldagain{Now, we take care of}\newagain{To handle} the time period $[0, T_t]$, \old{ preceding the first detection.
The system needs to be maintained in the domain of the main controller.}
\oldagain{One simply needs to}\newagain{we need to find} a static open-loop controller $C_p$ that solves $(S_{\tau}^{(X_p)}, \psi_p)$, where
\begin{equation}
	\nonumber
	\psi_p := 
		\begin{cases} 
			\nonumber
			\Safe_{[0,T_t]}(H_{\tau}( X_D )), 	& \mbox{if } X_p \subseteq  \mathcal{D}(\sym{C})\\ 
			
			\nonumber
			\Reach_{[0,T_t]}(H_{\tau}(X_D )), 	& \mbox{if } \mathcal{D}(\sym{C}) \subset  X_p\\ 
		\end{cases},
\end{equation}
and $X_D := \underset{\sym{x} \in \mathcal{D}(\sym{C})}{\bigcup}  \sym{x}$.
$C_p$ encapsulates at least one control input sequence $\tilde{u}_p \in \sym{U}^{T_t}$ that results in an output sequence $\tilde{y}_p \in Y^{T_t}_{\tau}$ such that $\tilde{y}_p \in \psi_p$.
One direct approach to find $\tilde{u}_p$ is via an exhaustive search in $\sym{U}^{T_t}$.

\section{Case Studies}
\label{sec_examples}
We provide different examples to demonstrate the practicality and applicability of the presented methodologies.
Implementations of all examples are done using available open-source toolboxes and some customized \CPP{} programs developed for each methodology.
All closed-loop simulations of refined controllers are done in \MATLAB{}.
We use a PC (Intel Xeon E5-1620 3.5 GHz and 32 GB RAM) for all the examples.

In all of the examples, given a concrete system $S_\tau$\old{, as defined in \eqref{eq_concrete_output_system}}, we construct a symbolic model $S_q$\old{, as defined in \eqref{eq_abstract_output_system}}.
We use tool \SCOTS{} \cite{SCOTS} to construct\old{ the underlying state-based system} $S_{q,X_q}$.
\SCOTS{} can only construct $S_{q,X_q}$ with an FRR $Q$ in the form:
\begin{align}
	\nonumber
	Q &:= \{ (x_\tau,x_q) \;\vert\; x_\tau \in X_\tau \cap x_q \land x_q \in X_q \},
\end{align}
where $X_q$ is a partition on $X_\tau$ constructed by a uniform quantization parameter $\eta \in \R^n$.
Declaring $\eta$ is sufficient to define $X_q$ and $Q$.
$X_q$ is a set of polytopes of identical shapes forming a partition on $X_\tau$.
This is a limited structure in constructing $\sym{S}$ that we must comply with\old{ in all of the examples}.
Another restriction imposed by \SCOTS{} is the need to use easily invertible output maps $h$ such that $H_q^{-1}(y_q)$, $y_q \in Y_q$, complies with the hyper-rectangular structure of $X_q$ needed by \SCOTS.
\old{
In all examples, the bounds and quantization parameters for states and inputs sets used to construct $S_{q,X_q}$ are selected based on initial simulations of $S_{\tau}$ using \MATLAB.}

\subsection{Output-Feedback Symbolic Control using Games of Imperfect Information}
We consider one example to illustrate the methodology presented in Section \ref{meth_game_based}.
In this example, after constructing $S_q$, tool \ALPAGA{} \cite{ALPAGA} is used to construct the knowledge-based game $\sym{S}^\mathcal{K}$ and synthesize a winning strategy $\bar{\mathcal{C}}^\mathcal{K}$.
We refine the strategy as previously depicted in Fig. \ref{fig_games_method}.

Consider the following dynamics of a DC motor:
\begin{equation}
	\nonumber
	\begin{bmatrix} 
	\dot{x_1}\\
	\dot{x_2}\\
	\dot{x_3}
	\end{bmatrix} 
	=  
	\begin{bmatrix} 
	\frac{-R}{L}& 0& \frac{-K}{L}\\ 
	0& 0& 1\\
	\frac{K}{J}& 0& \frac{-b}{J}
	\end{bmatrix} 
	\begin{bmatrix} 
	x_1\\
	x_2\\
	x_3
	\end{bmatrix} 	
	+
	\begin{bmatrix} 
	\frac{1}{L}\\
	0\\
	0
	\end{bmatrix} 	
	\upsilon,
\end{equation}
where $x_1$ is the armature current, $x_2$ is the rotation angle of the rotor, $x_3$ is the angular velocity of the rotor, $\upsilon$ is the input voltage, $L := 5\times 10^{-2}$ is the electric inductance of the motor coil, $R := 5$ is the resistance of the motor coil, $J := 5\times 10^{-4}$ is the moment of inertia of the rotor, $b := 1\times 10^{-2}$ is the viscous friction constant, and $K := 0.1$ is both the torque and the back EMF constants.
We consider a state set $X_\tau := [-0.6,0.6]\times[-0.3,0.3]\times[-4.8,4.8]$ and an input set $U_\tau := [-4.25,4.25]$.
One sensor is attached to the motor's rotor and it can measure $x_2$. 
Hence, the output is as follows:
\begin{equation}
	\nonumber
	y := 
	\begin{bmatrix} 
	0& 1& 0
	\end{bmatrix} 	
	\begin{bmatrix} 
	x_1\\
	x_2\\
	x_3
	\end{bmatrix},
\end{equation}
and, consequently, $Y_\tau := [-0.3,0.3]$.
We consider a reachability specification with a target set $T := [0.18, 0.3]$.

To construct $S_q$, we consider an abstract output spaces $\sym{Y} := \{y_0,y_2,y_3,\cdots,y_{30}\}$ that forces a partition on $Y_\tau$.
Here, each $y_q \in Y_q$ represents one subset in $Y_\tau$ from 31 subsets by dividing $Y_\tau$ equally using a quantization parameter $0.02$.
More precisely, we use an OFRR:
\begin{equation} 
	\nonumber
	Z := \{ (y_\tau, y_q) \in Y_\tau \times Y_q \;\vert\; y_q = y_{ \lfloor (y_\tau+0.3)/0.02 \rfloor }\}.
\end{equation}
With such a $Y_q$, the abstract specification is to synthesize a controller to reach any of the symbolic outputs $y_{24}, y_{25}, \cdots, y_{30}$.
To construct $S_{q,X_q}$, we use the following parameters in \SCOTS: a state quantization vector $(0.3,0.02,1.6)$, an input quantization parameter $0.75$, and a sampling time $\tau := 0.05$ seconds.
\SCOTS{} constructs $S_{q,X_q}$ in 2 seconds with $X_q$ having $1085$ elements (each representing a hyper-rectangle in $X_\tau$) and $S_{q,X_q}$ having $88501$ transitions.
We then define $\sym{H}$ as follows:
\begin{equation} 
	\nonumber
	\sym{H}((x_{q,1}, x_{q,2}, x_{q,3})) := y_{ \lfloor (x_{q,2}+0.3)/0.02 \rfloor},
\end{equation}
which satisfies condition \eqref{eq_alignment_preservation_condition}.
We pass $\sym{S}$ to \ALPAGA{} which takes around $24$ hours to construct $S_q^{\mathcal{K}}$ and synthesize $\bar{C}^{\mathcal{K}}$, which is then refined as discussed in Fig. \ref{fig_games_method}.

\begin{figure}
	\centering
	\includegraphics[width=0.6\textwidth]{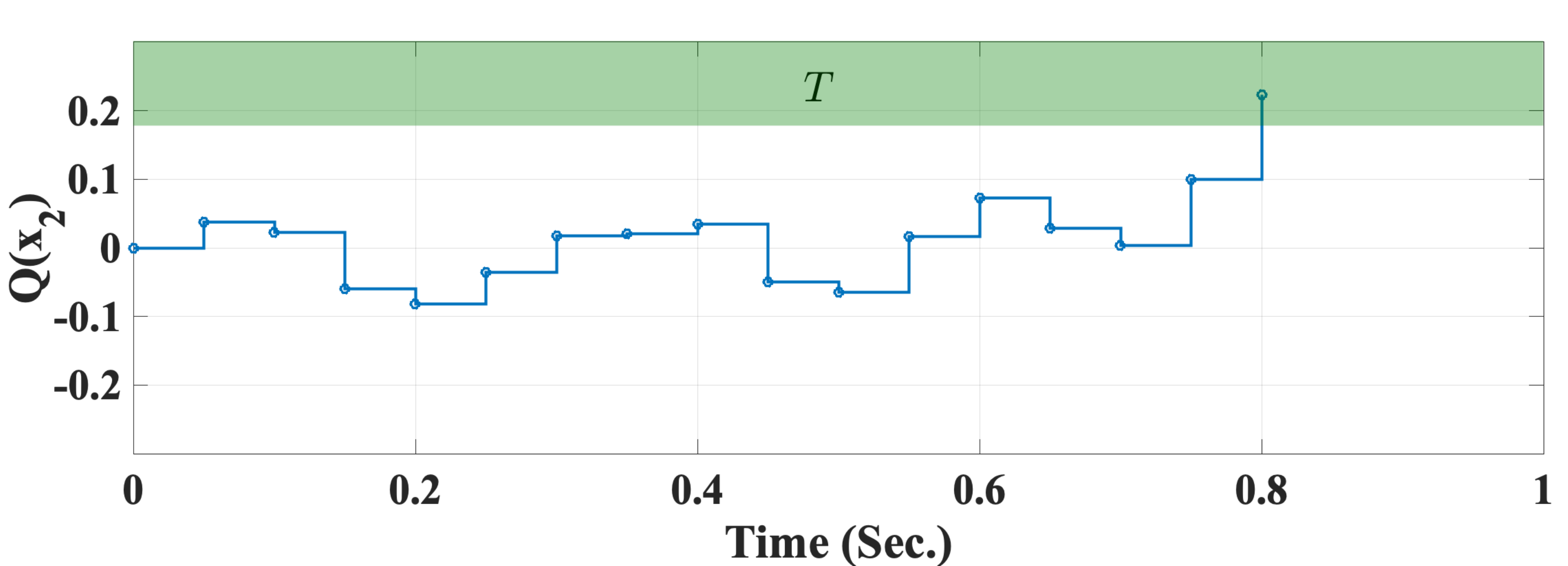}
	\caption{The output of the DC motor example.}	
	\label{fig_ex_game_dcmotor}
\end{figure}

The closed-loop behavior is simulated in \MATLAB{} and the output is depicted in Fig. \ref{fig_ex_game_dcmotor}.
The target region is highlighted with a green rectangle.
The actual initial state of the system is set to $(0,0,0)$, which is of course unknown to the controller.

\subsection{Output-Feedback Symbolic Control using Observers}
\label{ex_observer_dint}
As an example for the methodology presented in Section \ref{mthd_concrete_domain_observers}, consider the double-integrator model:
\begin{equation}
	\label{eq_dcmotor} 
	\begin{bmatrix} 
		\dot{x_1}\\
		\dot{x_2}
	\end{bmatrix} 
	=  
	\begin{bmatrix} 
		0& 1\\ 
		0& 0
	\end{bmatrix} 
	\begin{bmatrix} 
		x_1\\
		x_2
	\end{bmatrix} 	
	+
	\begin{bmatrix} 
		0\\
		1
	\end{bmatrix} 	
	\upsilon,
\end{equation}
where $(x_1,x_2) \in [-1,1]\times[-5,5]$, and $\upsilon \in [-10, 10]$.
The output of the system is seen through a single sensor monitoring $x_1$, i.e., $y = x_1$.
We consider the following LTL specification:
\begin{equation}
	\nonumber
	\psi_\tau = 
	\square\Diamond(\textsf{Target}1) 
	\wedge 
	\square\Diamond(\textsf{Target}2),
\end{equation}
where \new{$\Diamond(T)$ denotes the reachability requirement that the output of $\sym{S}$ visits, at least once, some elements in $T$,} $\textsf{Target}1 := [0.65, 1.0]$ and $\textsf{Target}2 := [-1,-0.65]$ are two subsets of $Y_\tau := [-1,1]$.

We first design an observer for the system.
We choose a precision value of $\epsilon := 0.001$ and design a Luenberger observer using pole placement.
It is then embedded in an observer system $\mathcal{O}$ that fulfills condition \eqref{eq_meth1_bounded_observ}.
System $\mathcal{O}$ is needed in order to refine the designed controller as depicted in Fig. \ref{fig_observer_method}.

To construct $S_q$, we set $Y_q := \{ y_0, y_1, \cdots, y_{50} \}$ forcing a partition on $Y_\tau$ such that each $y_q \in Y_q$ represents one subset of $Y_\tau$ from 51 subsets by dividing $Y_\tau$ equally using a quantization parameter $0.04$.
More precisely, use an OFRR:
\begin{equation} 
	\nonumber
	Z := \{ (y_\tau, y_q) \in Y_\tau \times Y_q \;\vert\; y_q = y_{ \lfloor (y_\tau + 1)/0.04 \rfloor } \}.
\end{equation}
Then, we use \SCOTS{} to construct $S_{q,X_q}$ with a sampling time $\tau := 0.05$ seconds, a state quantization vector $(0.04, 0.01)$, and an input quantization parameter $1.0$.
Error $\epsilon$ is used as a state error parameter in \SCOTS{} to emulate the inflation discussed in Subsection \ref{ssec_observers_synth_refine}.
\SCOTS{} constructs $S_{q,X_q}$ in $39$ seconds and it has $5.59325 \times 10^{7}$ transitions.
We then have an output map defined as follows:
$\sym{H}((x_{q,1}, x_{q,2})) := y_{ \lfloor (x_{q,1}+1)/0.04 \rfloor}$,
which satisfies condition \eqref{eq_alignment_preservation_condition}.
With the above setup, we can use the results of the observer-based methodology and refine any synthesized controller for $S_{q,X_q}$ using $\mathcal{O}$ and $Q$.

We continue with controller synthesis and refinement.
Since \SCOTS{} requires specifications over symbolic states, the corresponding symbolic target state sets are computed by $Q(H_\tau^{-1}(\textsf{Target}1))$ and $Q(H_\tau^{-1}(\textsf{Target}2))$, respectively.
The controller is synthesized in $24$ seconds.
The set of possible control-actions for the first sampling period are identified as discussed in Subsection \ref{ssec_observers_blind_period}.
The input $\upsilon := 0$ is selected for the first sampling period.

\begin{figure}	
	\centering
	\includegraphics[width=0.8\textwidth]{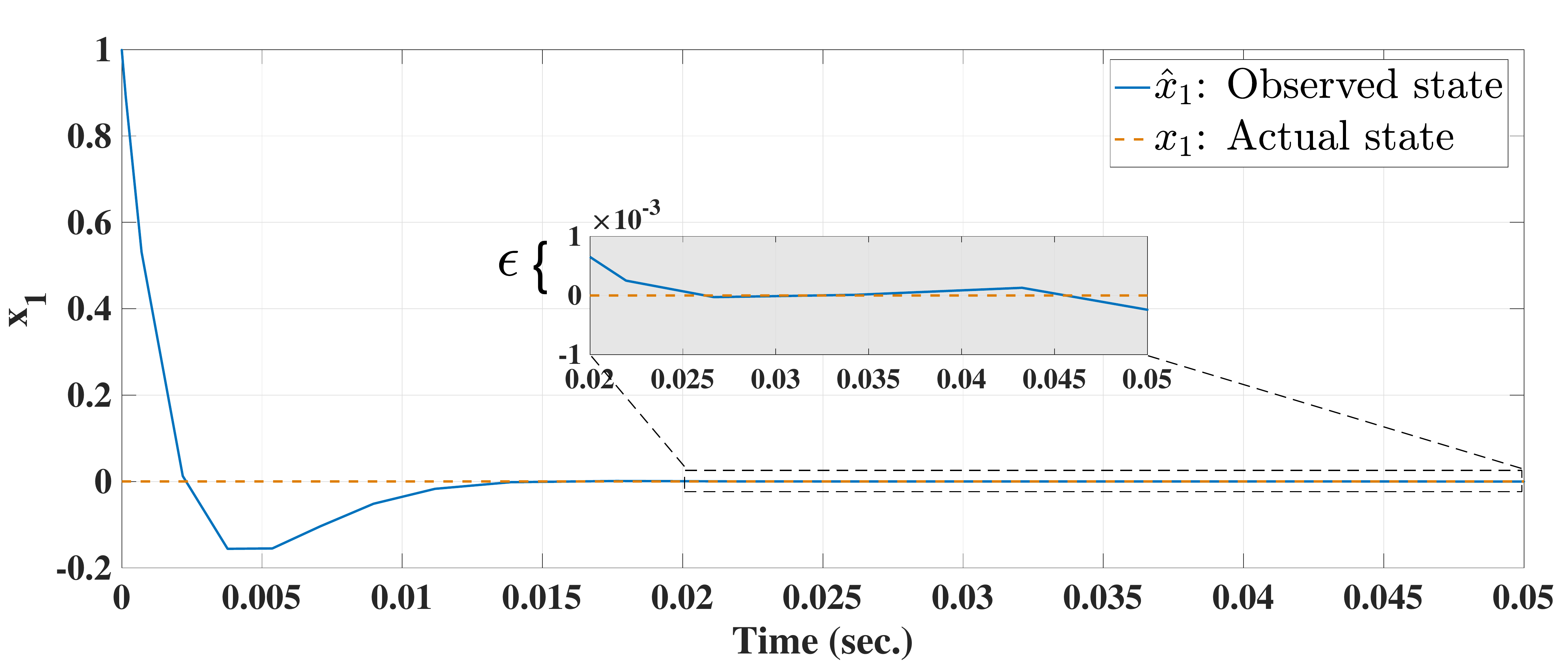}
	\caption{
		State variable $x_1$ and its observed version $\hat{x}_1$ of the double-integrator example during the first sampling period.
	}	
	\label{fig_ex_observer_dint_x_xhat_first_tau}
\end{figure}

\begin{figure}	
	\centering
	\includegraphics[width=0.7\textwidth]{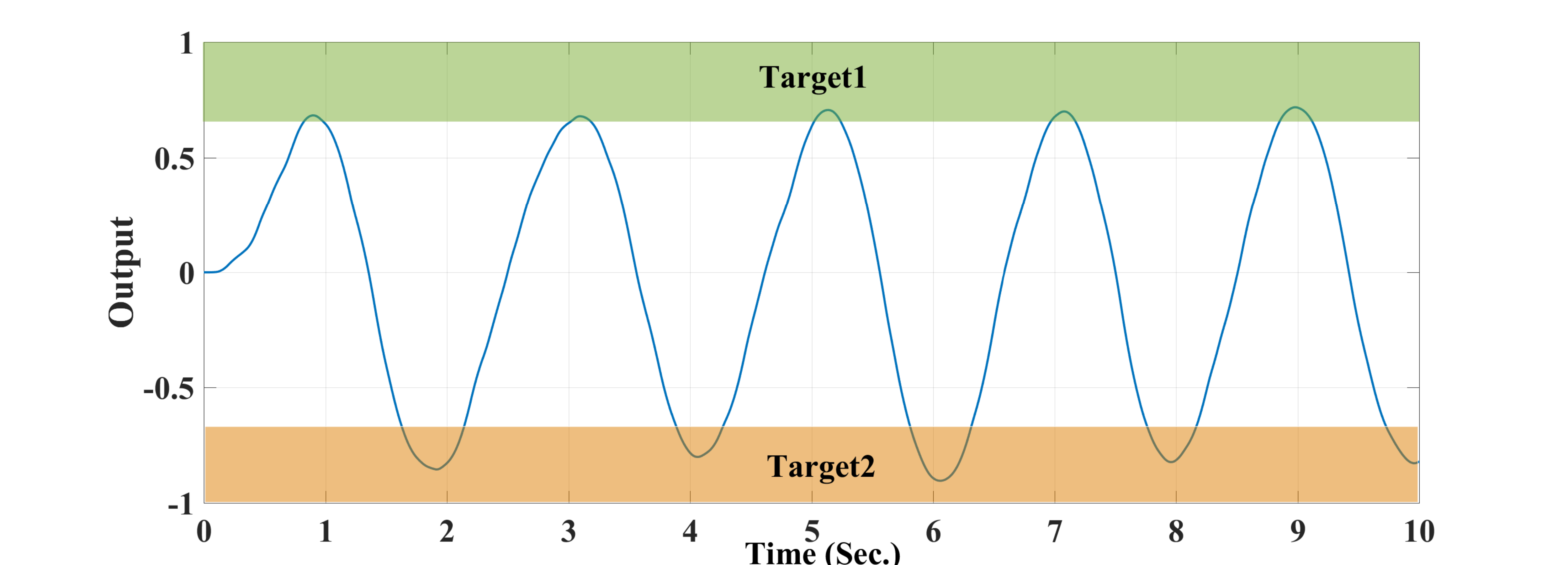}
	\caption{
		Output of the double-integrator example.
	}	
	\label{fig_ex_observer_dint_states}
\end{figure}

\begin{figure}
	\centering
	\includegraphics[width=0.8\textwidth]{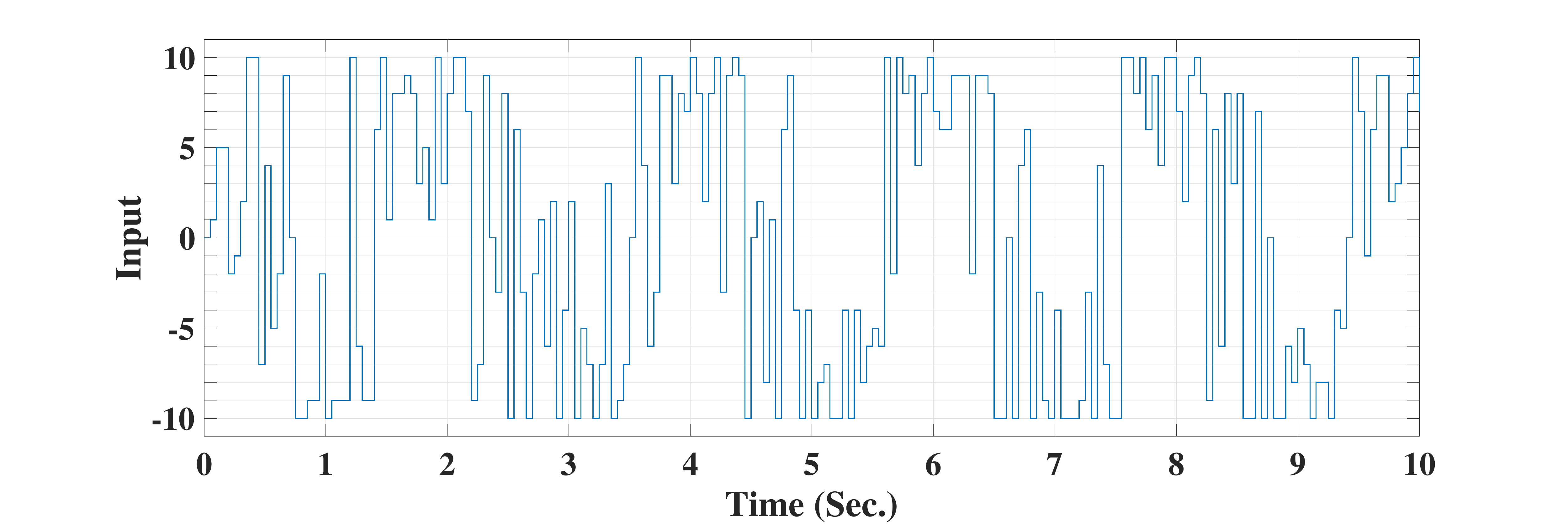}
	\caption{
		Control inputs applied to the double-integrator system during the closed-loop simulation.
	}	
	\label{fig_ex_observer_dint_input}
\end{figure}

We simulate the closed-loop in \MATLAB{} with $(0,0)$ and $(1,1)$ as initial states of the system and observer, respectively.
At the first sampling period, the controller applies input $\upsilon := 0$ to keep the system in the controller's domain.
From the second sampling period, we switch to the symbolic controller.
Figure \ref{fig_ex_observer_dint_states} depicts the output $y$ and Fig. \ref{fig_ex_observer_dint_input} depicts the applied inputs.

\subsection{Output-Feedback Symbolic Control using Detectors}
\label{ex_detector_pendulum}
Now, we provide an example to illustrate the methodology presented in Section \ref{meth_detector}. Consider a pendulum system \cite{G_POLA_Automatica_2008}:
\begin{equation} 
	\nonumber
	\begin{bmatrix} 
		\dot{x}_1\\ \dot{x}_2\\ 
	\end{bmatrix} 
	=  	
	\begin{bmatrix} 
		0 & x_2\\ 
		-\frac{g}{l}\sin(x_1) & -\frac{k}{m}x_2\\ 
	\end{bmatrix} 		
	+
	\begin{bmatrix} 
		0\\ 
		1\\ 
	\end{bmatrix} 			
	u,\;\;
	y
	=
	\begin{bmatrix} 
		1 & 0\\ 
	\end{bmatrix}
	\begin{bmatrix} 
		x_1\\ x_2\\ 
	\end{bmatrix},
\end{equation}
where $x_1 \in [-1,1]$ is the angular position, $x_2 \in [-1,1]$ is the angular velocity, $u \in [-1.5,1.5]$ is the input torque, $g := 9.8$ is the gravitational acceleration constant, $l := 5$ is the length of the pendulum's massless rod, $m := 0.5$ is a mass attached to the rod, $k := 3$ is the friction's coefficient, and $y \in [-1,1]$ is the measured angular position.
We consider designing a symbolic controller to enforce the angle of the rod to infinitely alternate between two regions $\theta_1 := [0.3,0.4]$ and $\theta_2 := [-0.4,-0.3]$.
When it reaches one region, the pendulum should hold for 10 consequent time steps.

To construct $S_q$, we set $Y_q := \{ y_0, y_1, \cdots, y_{50} \}$ forcing a partition on $Y_\tau$ such that each $y_q \in Y_q$ represents one subset in $Y_\tau$ from 51 subsets by dividing $Y_\tau$ equally using a quantization parameter $0.04$.
More precisely, we use an OFRR:
\begin{equation} 
	\nonumber
	Z := \{ (y_\tau, y_q) \in Y_\tau \times Y_q \;\vert\; y_q = y_{ \lfloor (y_\tau + 1)/0.04 \rfloor } \}.
\end{equation}
$S_{q,X_q}$ is constructed using\old{ \SCOTS{} with} the following parameters: state quantization vector $(0.4, 0.4)$, input quantization parameter $0.15$, and a sampling time $2$ seconds.
The resulting $S_{q,X_q}$ has 25 states and 525 transitions.
We then have an output map defined as follows:
$\sym{H}((x_{q,1}, x_{q,2})) := y_{ \lfloor (x_{q,1}+1)/0.04 \rfloor}$,
which satisfies condition \eqref{eq_alignment_preservation_condition}.
\old{With the above setup, }We \old{can}\new{then} use the results\old{ of the detector-based methodology} \new{from Section \ref{meth_detector}} and refine any synthesized controller for $S_q$.

\begin{figure}
	\centering
	\includegraphics[width=0.8\textwidth]{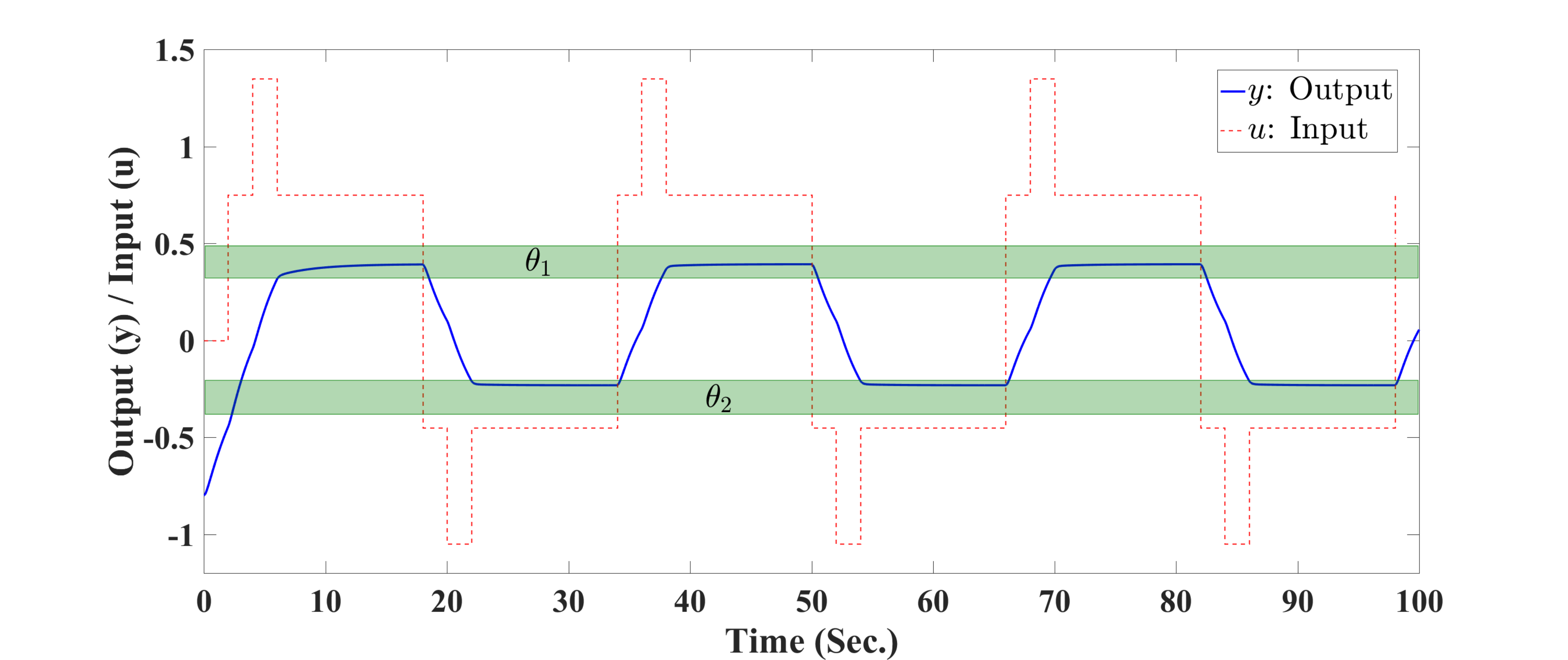}
	\caption{Closed-loop simulation of Pendulum example.}
	\label{fig_ex_detector_pendulum_io}
\end{figure}

We implemented Algorithm \ref{alg1:detectability_NFTS} in \CPP{} and ran it with $S_q$ as input.
NFA $\mathcal{A}$ has 60 states and 1485 transitions.
System $\sym{S}$ is detectable with $T_t = 1$.
A controller is synthesized using \SCOTS{} and map $H_{\tau}^{-1}$ is used to construct a state-based specification.
The controller is refined using $Z$ and the detector\old{ as depicted in Fig. \ref{fig_detector_method}}.
\old{We also identify the set of input sequences that can be applied in open-loop before the detector is able to detect the state of the symbolic model.
We continuously apply input $0$ until the state is detected.
We simulate the closed-loop with an initial state $(-0.8, 0)$.
Once the detector identifies current state, we switch to the synthesized controller to enforce the specifications.}
A closed-loop simulation is depicted in Fig. \ref{fig_ex_detector_pendulum_io}.

\section{Related Works}
\label{sec_related_works}
 
\old{There already exist results in the computer science literatures for formal synthesis based on output but only for finite-state models.
These results include controller synthesis for partial-observable parity games \cite{ARNOLD_2003}, counter-example based optimal schedulers under partial information \cite{Giro_2012}, omega-regular two player games with imperfect information \cite{Chatterjee2006_omegareg_games}, hidden Markov models \cite{Zhang_2005}, mixed heuristics and model checking for the synthesis of controllers based on outputs \cite{Arxiv_ChatterjeeCGK14}, and more recently, by approximately solving dynamic programs for partially observable Markov decision processes \cite{Lesser_2017}, or by utilizing estimators of finite systems \cite{Ehlers_HSCC15}.}
 
The work in \cite{Abate_2015_OutFeedback} provides a symbolic control approach based on outputs.
It is limited to partially observable linear time-invariant systems, as long as the system is detectable and stabilizable.
\old{Additionally, it relies on approximate simulation relations which, as argued by several examples in \cite[Section IV]{G.Reissig_etal_FRR_TAC}, have several drawbacks.}
Some extensions are made in \cite{Lesser_2015} for probabilistic safety specifications and in \cite{Lesser_2016} for nonlinear systems.
The latter is limited to a class of feedback-linearizable systems and the results are limited to safety.
 
The work in \cite{Masashi_NonLinearHybSystems2018} proposes designing symbolic output-feedback controllers for control systems.
It designs observers induced by abstract systems and obtain output-feedback controllers similar to the methodology we presented in Section \ref{meth_detector}.
The authors, unlike our approach, require the availability of a controller for the abstract system when the state of the control system is fully measured.
Then, they reduce the controller to work with the original system with the designed observer.
 
 
\old{The closest works to this article are in \cite{GPola_CDC17,PolaOutSymControl2019,ApazaObserverBased,RupakOutSymControl2020}.}
In \cite{GPola_CDC17,PolaOutSymControl2019}, the authors use \emph{state-based} strong alternating approximate simulation relations to relate concrete systems with their abstractions.
They make sure that a partition constructed on the output space imposes a partition on the state space, which allows designing output-based controllers using \emph{state-based} symbolic models.
\newagain{The work in \cite{PolaOutSymControl2019} is different from ours in three main directions: (1) our work introduces OFRRs as general relations between the outputs of symbolic models and original systems, (2) we utilize FRRs which avoid the drawbacks of approximate alternating simulation relations (see \cite[Section IV]{G.Reissig_etal_FRR_TAC} for a comparison between both types of relations), and (3) we introduce multiple practical methodologies that realize the framework we introduced; in Sections \ref{meth_game_based}, \ref{mthd_concrete_domain_observers}, and \ref{meth_detector}.}
In \cite{ApazaObserverBased}, the authors design observers for original systems.
Then, the observed \emph{state-based} systems are related, via FRRs, to \emph{state-based} symbolic models that are used for controller synthesis.
Unlike our work, the behavioral inclusion from original closed-loop to abstract closed-loop is shown in state-based setting.
Also, the specifications are given over the states set.
In \cite{RupakOutSymControl2020}, the authors provide an extension to FRR to ensure that controllers designed for \emph{state-based symbolic models} can be refined to work for output-based concrete systems.
Abstractions are designed using a modified version of the knowledge-based algorithm\old{ traditionally used for games of imperfect information}\new{ (a.k.a. KAM)}.
Unfortunately, the authors can not decide whether \oldagain{the}\newagain{a} correct abstraction is constructed or not unless a controller is synthesized which requires to iteratively run the algorithm.
\new{KAM needs to be stopped once an upper bound for the number of iterations is reached.}
\new{Although Algorithm \ref{alg1:detectability_NFTS} is more restrictive in the sense that KAM can produce an abstraction for a symbolic model that is not detectable, it is more predictable since it always terminates.}
\new{Additionally, Algorithm \ref{alg1:detectability_NFTS} runs in polynomial time, while the KAM algorithm runs in exponential time. Hence, although KAM algorithm can work for undetectable systems, Algorithm \ref{alg1:detectability_NFTS} is significantly more efficient for detectable systems.}
\new{Having both algorithms available to the designer of symbolic controllers offers a trade-off between decidability and applicability.}

\oldagain{In comparison with these related works, t}\newagain{T}he main contributions of this work are:
\begin{itemize}
	\item[1.]
	\newagain{OFRRs are introduced as extensions to FRRs allowing abstractions to be constructed by quantizing the state and output sets of concrete systems, such that the output quantization respects the state quantization in the sense that every quantized state belongs to one quantized output.
	Symbolic controllers of output-based symbolic models can be refined to work for output-based concrete systems.}

	\item[2.]
	OFRRs and the results following them in Section \ref{sec_outfb_sym_control} serve as a general framework to host different methodologies of output-feedback symbolic control.
	\oldagain{While each of the discussed works introduces one specific methodology of output-feedback symbolic control, our general framework can host many other methodologies for output-feedback symbolic control.}
	\old{For example, in Section \ref{meth_game_based} we provide a methodology for pure output-based controller synthesis.
	Unlike our work, the authors in \cite{Abate_2015_OutFeedback,GPola_CDC17,PolaOutSymControl2019,ApazaObserverBased,RupakOutSymControl2020} rely on constructing controllers using \emph{state-based} symbolic models.}
	
	\item[3.]
	\old{We showed that similar methodologies using observers  \cite{Abate_2015_OutFeedback,GPola_CDC17} or games of imperfect information \cite{RupakOutSymControl2020} can be easily embedded in our framework (see Sections \ref{mthd_concrete_domain_observers} and \ref{meth_game_based}).}
	We introduced three example methodologies (Sections \ref{meth_game_based}, \ref{mthd_concrete_domain_observers}, and \ref{meth_detector}) to synthesize and refine output-based symbolic controllers for output-based systems.
\end{itemize}

\section{Conclusion}
\begin{table}
	\centering
	\caption{Requirements of the presented methodologies.}
	\label{tbl_methodologies_and_their_complexities}
	\begin{tabular}{l|l|l}
		\textbf{Methodology}& \textbf{Assumptions}	& \textbf{Refined controllers} 		\\ \cline{1-3}
		2-player games 	& None. 				& $Z$ + symbolic model 		\\
		Observer based 	   	& $S_\tau$ is observable& $Z$ + observer 			\\ 
		Detector based		& $S_q$ is detectable 	& $Z$ + detector 			\\ \cline{1-3}
	\end{tabular}
	\vspace{-0.6cm}
\end{table}

We have shown that symbolic control can be extended to work with output-based systems.
OFRR are introduced as tools to relate systems based on their outputs.
\newagain{They allow symbolic models to be constructed by quantizing the state and output sets of concrete systems, such that the output quantization respects the state quantization.}
\newagain{Consequently, this}\oldagain{which} allows refining symbolic controllers \newagain{designed based on the outputs of symbolic models} to work with the outputs of original systems.
Three example methodologies for output-feedback symbolic control \newagain{based on detectors for symbolic models} were also introduced.
Their assumptions and requirements are highlighted in Table \ref{tbl_methodologies_and_their_complexities}.


\bibliographystyle{IEEEtran}
\bibliography{refdb}
\old{
\begin{IEEEbiography}
	[{\includegraphics[width=1.0in]{bio/khaled.jpg}}]
	{Mahmoud Khaled}
	is a lecturer in the Department of Computer and Systems Engineering, Minia University, Egypt. He received a B.Sc. in Computer and Systems Engineering (2009), an M.Sc. in Electrical Engineering (2014) from Minia University, Egypt, and a PhD in Electrical Engineering and Computer Science (2021) from Technical University of Munich, Germany. His current interests include formal methods for automated synthesis of safety-critical Cyber-physical systems, embedded control systems and networked control systems.
\end{IEEEbiography}

\begin{IEEEbiography}
	[{\includegraphics[width=1in]{bio/zhang.png}}]
	{Kuize Zhang} (M'14-SM'17) received B.S. and Ph.D. degrees in Mathematics and Control Science and Engineering from Harbin Engineering University, China, in 2009 and 2014, respectively. He is currently a Humboldt postdoc at Control Systems Group, Technical University of Berlin, Germany. He held postdoc positions at KTH Royal Institute of Technology, Sweden (Sep. 2017– Apr. 2020), Technical University of Munich, Germany (Sep. 2016–Aug. 2017), the Chinese Academy of Sciences (Jan. 2015–Aug. 2016), Nanyang Technological University, Singapore (Oct. 2013–Oct. 2014), and University of Turku, Finland (Sep. 2012–Sep. 2013). His research interests include fundamental topics in Boolean networks, discrete-event systems (finite automata, Petri nets, weighted automata), etc. He co-authored one monograph.
\end{IEEEbiography}

\begin{IEEEbiography}
	[{\includegraphics[width=1in]{bio/zamani.jpg}}]
	{Majid Zamani}
	is an Assistant Professor in the Computer Science Department, University of Colorado Boulder, USA. He received a B.Sc. degree in Electrical Engineering in 2005 from Isfahan University of Technology, Iran, an M.Sc.
	degree in Electrical Engineering in 2007 from Sharif University of Technology, Iran, an MA degree in Mathematics and a Ph.D. degree in Electrical Engineering both in 2012 from University of California, Los Angeles, USA. Between September 2012 and December 2013 he was a postdoctoral researcher at Delft University of Technology, Netherlands. From May 2014 to January 2019 he was an Assistant Professor in the Department of Electrical and Computer Engineering at the Technical University of Munich, Germany. He received an ERC starting grant award from the European Research Council in 2018. His research interests include verification and control of hybrid systems, embedded control software synthesis, networked control systems, and incremental properties of nonlinear control systems.	
\end{IEEEbiography}
}

\end{document}